\documentclass{article}

\PassOptionsToPackage{square,sort,comma,numbers,compress}{natbib} %

\usepackage[main,preprint]{neurips_mod}

\usepackage[utf8]{inputenc} %
\usepackage[T1]{fontenc}    %
\usepackage{hyperref}       %
\usepackage{url}            %
\usepackage{booktabs}       %
\usepackage{amsfonts}       %
\usepackage{nicefrac}       %
\usepackage{microtype}      %
\usepackage{xcolor}         %
\hypersetup{colorlinks, linkcolor={green!40!black}, citecolor={red!50!black}, urlcolor={blue!60!black},}

\usepackage[utf8]{inputenc}
\usepackage[at]{easylist}
\usepackage{fullpage}

\usepackage{amsmath,amsfonts}
\usepackage{amssymb}
\usepackage{amsthm}
\usepackage{xcolor}
\usepackage{mathrsfs}
\usepackage{titlesec}
\usepackage{xspace}
\usepackage{hyperref}
\usepackage{cleveref}
\usepackage{graphicx}
\usepackage{thmtools, thm-restate}

\graphicspath{{Figures/}}

\usepackage{algorithm}
\usepackage{algorithmicx}
\usepackage[noend]{algpseudocode}

\usepackage{booktabs}
\usepackage{multirow}
\usepackage{makecell}

\usepackage{pifont}%
\definecolor{checkcolor}{HTML}{005AB5}
\definecolor{crosscolor}{HTML}{DC3220}
\newcommand{\cmark}{\textcolor{checkcolor}{\ding{51}}}%
\newcommand{\xmark}{\textcolor{crosscolor}{\ding{55}}}%
\usepackage[skip=1pt,font=footnotesize]{subcaption}
\usepackage{dsfont}
\newcommand{\argmin}{\operatorname*{arg\,min}}
\newcommand{\argmax}{\operatorname*{arg\,max}}

\newtheorem{theorem}{Theorem}[section]

\newtheorem{proposition}[theorem]{Proposition}

\newtheorem{definition}[theorem]{Definition}

\newcommand{\mathify}[1]{\ensuremath{#1}\xspace}

\newcommand{\adv}{\mathify{\mathcal{A}}} %
\newcommand{\mech}{\mathify{\mathcal{M}}} %
\newcommand{\ds}{\mathify{\mathbf{x}}} %

\newcommand{\mechout}{\mathify{a}}
\newcommand{\Mechout}{\mathify{A}}
\newcommand{\dsfixed}{\mathify{\mathbf{x}_{fixed}}} %
\newcommand{\dd}{\mathify{\mathcal{D}}} %
\newcommand{\class}{\mathify{\mathcal{C}}} 
\newcommand{\E}{\mathify{\mathbb{E}}} %
\newcommand{\I}{\mathify{\mathbb{I}}} %
\newcommand{\targets}{\mathify{\mathbf{z}}} %
\newcommand{\guesses}{\mathify{\targets}}
\newcommand{\cZ}{\mathify{\mathcal{Z}}} %
\newcommand{\cX}{\mathify{\mathcal{X}}} %
\newcommand{\R}{\mathify{\mathbb{R}}} %
\newcommand{\eps}{\varepsilon} %
\newcommand{\bern}{\mathsf{Bernoulli}}

\newcommand{\metric}{\mathify{\mathcal{L}}} %
\newcommand{\sd}{\mathify{\preceq_{s.d.}}}

\title{A Unified Framework for Adversary-Aware Differential Privacy Bounds }

\author{%
  Marika Swanberg\thanks{Equal contribution}\\Google Research
  \And
  Meenatchi Sundaram Muthu Selva Annamalai\footnotemark[1]\\University College London
  \And
  Jamie Hayes\\Google DeepMind
  \And Borja Balle\thanks{Work done while at Google DeepMind}\\Optiak
  \And Adam Smith\\Boston University
}

\begin{document}

\maketitle

\begin{abstract}
Differential Privacy (DP) bounds the privacy leakage of a mechanism against worst-case membership inference, but the precise tradeoff between complex adversarial models and DP protections remains poorly understood. In this paper, we present a unified framework that generalizes the patchwork of existing bounds across membership inference, attribute inference, and data reconstruction attacks. 
Crucially, our framework is the first to evaluate attacks that target multiple individuals simultaneously and measure success beyond exact matches under a single cohesive bound. Our bounds capture this broad family of previously unexplored attack settings by relying solely on the privacy parameters and the adversary's baseline success rate (i.e. its prior without access to the mechanism's output). To illustrate this, we compare our high-probability guarantees to empirical attacks in two novel settings: extracting multiple non-uniform secrets (passwords and PII) from DP-finetuned language models, and reconstructing tabular data from noisy marginals. Ultimately, this framework provides a rigorous theoretical foundation to investigate the risk landscape of DP algorithms in new adversarial settings.

\end{abstract}

\section{Introduction}
Differential privacy (DP)~\cite{dwork2006calibrating} has emerged as a standard for designing mechanisms that process sensitive data in a variety of deployment contexts \cite{Apple17,GoogleDP,Google-FL-blog, Abowd2018TheUC,HodCanetti2025}. 
Unlike ad-hoc methods, DP provably bounds the amount of information that a mechanism reveals about any individual using the privacy parameter(s), with larger parameters typically giving weaker privacy guarantees.

A central challenge in DP deployments is how to interpret the privacy protections of the chosen privacy parameters in context. The DP guarantee naturally bounds the success of a worst-case membership inference (MI) adversary, who attempts to distinguish between the outputs of the mechanism on two inputs that differ on a single user's data~\cite{kairouz2015composition}. While there are numerous good reasons to use this threat model in \emph{defining} DP%
\footnote{To name a few: worst-case MI \emph{is} the correct threat model for some mechanisms like randomized response, the guarantee is dataset-agnostic, it implies protections against less-informed adversaries, it allows us to compare different mechanisms, etc.},
it provides little insight into \emph{accurately quantifying} the risk under more challenging adversarial modeling assumptions. Consider a DP deployment that uses $\eps = 1$. The DP bound tells us that worst-case MI succeeds with probability at most $0.73$\footnote{In general the probability of success against worst-case MIA for any $\eps$-DP mechanism is at most $\frac{e^\eps}{e^\eps + 1}$.}. Now suppose the primary concern for this deployment isn't membership inference but rather the total number of users' data that could be partially reconstructed, possibly using prior information about the user population. Until now, practitioners have not had any means to \emph{translate} the $0.73$ probability of worst-case MI into context-dependent privacy risks. This problem is magnified for privacy parameters that provide only ``vacuous guarantees'' against worst-case MI but which may empirically protect against the risks that are relevant to the deployment.

Practitioners in this situation often rely on empirical privacy attacks to gauge risk. However, this presents several challenges. First, empirical attacks merely establish lower bounds on privacy leakage; future, more sophisticated attacks could easily invalidate current risk assessments. Furthermore, state-of-the-art attacks are often highly tuned to specific mechanisms, hyperparameter settings, or data distributions, making them difficult to generalize. For example, even within DP-SGD, epsilons ranging from $10^2$ to $10^9$ have thwarted state-of-the-art reconstruction attacks \citep{carlini2019secret, ponomareva2022training, balle2022reconstructing}. Unless a practitioner's exact deployment and threat model mirrors a published experiment, empirical findings offer modest actionable guidance. Furthermore, as \citet{cummings2024attaxonomy} argue, many attack settings that are relevant to practitioners have not been studied yet. In the end, the attacks' findings are boiled down to general rules of thumb like the ``undocumented but still widely used goal for DP ML models of achieving an $\eps \leq 10$''~\cite{ponomareva2023dp}.

\begin{table*}[t]
    \begin{tabular}{l@{}c@{}cc@{}c@{}cc@{}cc@{}c}
        \toprule
        \multirow{2}{*}{\bf Reference} & \multicolumn{2}{c}{\bf Prior Knowledge} & \multicolumn{3}{c}{\bf Attack Goal} & \multicolumn{2}{c}{\bf No. of Targets} & \multicolumn{2}{c}{\bf Success Metric} \\
        \cmidrule(l){2-3}        \cmidrule(l){4-6}         \cmidrule(l){7-8}         \cmidrule(l){9-10}
        & {\bf\small Uniform~} & {\bf\small Non-Uniform} & {\bf\small MI~} & {\bf\small AI~} & {\bf\small Recon} & {\bf\small Single~} & {\bf\small Multiple} & {\bf\small Simple~} & {\bf\small Complex} \\
        \midrule
        \citet{balle2022reconstructing} & \cmark & \cmark & \cmark & \cmark & \cmark & \cmark & \xmark & \cmark & \cmark \\
        \citet{hayes2024bounding} & \cmark & \cmark & \cmark & \cmark & \cmark & \cmark & \xmark & \cmark & \cmark \\
        \citet{cummings2024attaxonomy} & \cmark & \cmark & \cmark & \cmark & \cmark & \cmark & \xmark & \cmark & \cmark \\
        \citet{kulynych2025unifying} & \cmark & \cmark & \cmark & \cmark & \cmark & \cmark & \xmark & \cmark & \cmark \\
        \citet{stock2022defending} & \cmark & \cmark & \cmark & \cmark & \cmark & \cmark & \xmark & \cmark & \xmark \\
        \citet{steinke2024privacy} & \cmark & \cmark & \cmark & \xmark & \xmark & \cmark & \cmark & \cmark & \xmark \\
        \citet{mahloujifar2024auditing} & \cmark & \xmark & \cmark & \xmark & \xmark & \cmark & \cmark & \cmark & \xmark \\
        \citet{cherubin2024closed} & \cmark & \xmark & \cmark & \cmark & \xmark & \cmark & \cmark & \cmark & \xmark \\
        \midrule
        {\bf Ours} & \cmark & \cmark & \cmark & \cmark & \cmark & \cmark & \cmark & \cmark & \cmark \\
        \bottomrule
    \end{tabular}
    \small
    \centering
    \vspace*{2mm}
    \caption{Comparing the applicability of our framework with prior work in various settings (see Section~\ref{sec:threat_model}).} %
    \label{tab:compare_prior_work}
\end{table*} \subsection{Our Contributions}
In this work, we formalize precise, adversary-aware DP upper bounds that quantify privacy risks across previously uncharacterized threat models. While prior bounds have largely treated membership inference, attribute inference, and exact reconstruction as isolated scenarios, our framework mathematically bridges these paradigms. Specifically, we provide the machinery to bound attacks along three simultaneous dimensions: the scale of the attack (targeting multiple individuals), the skew of the adversary’s prior knowledge (handling non-uniform distributions), and the strictness of the success criteria (accommodating approximate, rather than purely exact, data recovery). At a technical level, our bounds achieve this flexibility by depending exclusively on the privacy parameters and the adversary’s baseline probability of making a successful guess without access to the mechanism's output.
More specifically:

\paragraph{Unifying Previous Work.} Our bounds unify the patchwork of bounds focused on various settings of membership inference (MI) \cite{steinke2024privacy, mahloujifar2024auditing, cherubin2024closed}, attribute inference (AI)~\cite{cherubin2024closed}, and reconstruction \cite{stock2022defending, balle2022reconstructing, cummings2024attaxonomy, hayes2024bounding, cohen2025data, kulynych2025unifying}. Our bounds strictly generalize each of these works. A detailed discussion of related work and how our bounds bridge existing frameworks is provided in  Appendix~\ref{sec:bridging_attack}.

\paragraph{Covering New Settings.} 
The new bounds' generality allows them to be applied in previously unexplored settings. \Cref{tab:compare_prior_work} summarizes several dimensions in which one can compare existing bounds to ours; the table's terminology is detailed in \Cref{sec:threat_model}.

\paragraph{Empirical Case Studies.} To illustrate the versatility of our bounds and compare them to the success of concrete attacks, we study two specific settings that have been unexplored in prior work deriving upper bounds: extracting non-uniform secrets (passwords and PII) from DP-finetuned language models (Section~\ref{sec:novel_bounds}) and reconstructing tabular data from noisy marginals (Appendix~\ref{sec:exps}).
\section{Background}
In this section, we cover relevant background topics of DP and threat modeling terms. Informally, DP bounds the influence that any single user has on the output distribution of a mechanism. 
Throughout this paper, we assume the ``replace-one'' adjacency where exactly one record is replaced with another record.

\begin{definition}[Approximate DP~\cite{dwork2006calibrating}]
    \label{def:approxdp}
    A randomized mechanism $\mech : \mathcal{D} \rightarrow \mathcal{R}$ satisfies $(\eps, \delta)$-DP if, for any two adjacent datasets $\ds, \ds' \in \mathcal{D}$ and $S \subseteq \mathcal{R}$, it holds:
    \begin{equation*}
        \Pr[\mech(\ds) \in S]  \leq \mathrm{e}^\eps \Pr[\mech(\ds') \in S] + \delta
    \end{equation*}
\end{definition}

There are many other variants of DP, e.g., the pure DP variant where $\delta = 0$ and the more recent Rényi DP ($(\alpha, \gamma)$-RDP)~\cite{mironov2017renyi}, $f$-DP~\cite{dong2019gaussian}, and $\mu$-GDP~\cite{dong2019gaussian} variants, which characterize more complex relationships between the output distributions.

\subsection{Threat Modeling}
\label{sec:threat_model}
Several prior works have focused on systematizing and proposing a taxonomy of privacy attacks~\cite{rigaki2023survey,salem2023sok,cummings2024attaxonomy}. To categorize the privacy risks evaluated in \Cref{tab:compare_prior_work}, we adopt the taxonomy of \citet{cummings2024attaxonomy}. \textbf{Attack Goals} range from membership inference (MI) \citep{shokri2017membership}, to attribute inference (AI) \citep{ganju2018property}, and full data reconstruction (Recon) \citep{dinur2003revealing}. \textbf{Prior Knowledge} dictates whether the adversary assumes uniform probabilities \citep{mahloujifar2024auditing, steinke2024privacy} or exploits realistic, non-uniform data distributions \citep{dick2023confidence, gadotti2022pool}. \textbf{Number of Targets} distinguishes between isolating a single target versus simultaneously attacking multiple records \cite{cohen2022attacks,dick2023confidence,cohen2018linear,gadotti2019signal,annamalai2024linear,cherubin2024closed}. \textbf{Success Metrics} denote whether the attacker must achieve an exact match (Simple) or if they succeed via approximate reconstruction distances (Complex) \citep{balle2022reconstructing, carlini2021extracting, hayes2024bounding}.

\section{Our Privacy Leakage Framework}
\label{sec:our_framework}

In this section, we bound the success of an adversary in a \emph{generic attack game}. We begin by specifying the game and the setting it represents. Then in Theorem~\ref{thm:eps_bound}, we state our bound for pure DP mechanisms in this generic setting. After discussing multiple interpretations of the bound, we describe how to encode specific attack settings in the parameters of our framework. Next, we extend our bounds to approximate DP mechanisms and discuss the nuances that arise in this challenging setting. Lastly, we compare our bounds with the current state-of-the-art bounds to show how our bounds are tighter and more flexible than prior work.

\subsection{Generic Attack Game}

The \emph{generic attack game} in Algorithm~\ref{alg:dist_recon}  describes the family of attacks that our bounds cover. The game works as follows. First, $n$ targets $X$ are sampled from a given product distribution $\dd$. The targets $X$ are fed to the mechanism \mech, which produces an output $\mechout$. The adversary \adv uses the mechanism output $\mechout$, the data generating distribution $\dd$, and the mechanism specification \mech to produce $k$ attack attempts $\targets$. Lastly, the quality of the attack attempt \guesses is measured against the true dataset $X$ according to some success metric $\metric$.

\begin{algorithm}
        \caption{Generic Attack Game}
        \label{alg:dist_recon}
        \hspace*{\algorithmicindent} \textbf{Input:} Mechanism \mech, Adversary \adv, Data distribution $\dd$, Success metric \metric \\
        \hspace*{\algorithmicindent} \textbf{Output:} Attack success measure
        \begin{algorithmic}[1] %
                \State Sample targets $(X_1, \ldots, X_n) \sim \dd_1 \otimes \cdots \otimes \dd_n$
                \State Compute mechanism output $\mechout \gets \mech(X_1, \ldots, X_n)$
                \State Attack attempt $(\targets_1, \ldots, \targets_k) = \adv(\mechout, \dd, \mech)$
                \State \Return Attack success $\metric(X, \targets)$ 
        \end{algorithmic}
    \end{algorithm}

We call $\dd$ the adversary's \emph{prior distribution} over the records, and the prior $\dd_i$ for person $i \in [n]$ encodes the knowledge that the adversary has for that target person's data before seeing the mechanism output. More generally, our framework can be applied to whatever \emph{privacy unit} the mechanism \mech is designed to protect, such as, individual records, users, or households. Crucially, $X$ does not necessarily represent full data records. Instead, $\dd$ is defined over the unknown domain of the attack, whether that is membership bits, specific missing attributes, or full records. Our bounds handle arbitrary priors $\dd_1, \ldots, \dd_n$ for each target, as long as they are \emph{independent}. This independence is a fundamental modeling requirement for DP itself. If data records are correlated, the standard replace-one DP guarantee inherently breaks down, as the impact of one individual's information can be larger than what the mechanism was designed to protect~\cite{humphries2023investigatingmembershipinferenceattacks}.

The \emph{attack success metric} $\metric$ encodes the goal of the adversary. The metric $\metric: \cX \times \cZ \to \{0, \ldots, n\}$  takes a fixed list of targets $X$ of size $n$ and a list of \emph{guesses} or \emph{attack attempts} \guesses of size $k$, and reports the number of targets in $X$ that were successfully ``attacked'' by the adversary. Our bounds handle arbitrary success metrics, with the requirement that \metric is \emph{decomposable}, which we define next.

\begin{definition}[Decomposable metric]
Let $\cX$ and $\cZ$ be the data and attack domains, respectively. A metric $\metric: \cX \times \cZ \to \{0, \ldots, n\}$ is \emph{decomposable} into $\ell_1 \cdots \ell_n$ if 
\begin{equation*}
    \metric(\ds, \guesses) = \sum_{i \in [n]} \ell_i(\ds_i, \guesses).
\end{equation*}
 \end{definition}

In Section~\ref{sec:encoding_attacks} we discuss how to encode common attack goals---including MI, AI, and reconstruction---in the success metric. Although we do not explicitly state it, our generic attack game allows any number of fixed ``non-target'' records $\dsfixed$ to be input to the mechanism in addition to the sampled target records. We assume that the adversary has full knowledge of these records---and can even select them. As a result, the attack success is not measured with respect to the non-target records. We omit the fixed records in our attack game and bounds, but one can imagine the mechanism $\mech$ takes these as a hard-coded input.

\subsection{Bounding Pure DP Mechanisms}

We now present our bound for $\eps$-DP mechanisms. For a private dataset that is drawn from the prior distribution, and given any fixed mechanism output $\mechout$, we present a high probability bound on the attack success as a function of: (1) the mechanism's DP parameters, and (2) the adversary's success probability if they had only used prior information. Formally:

\begin{restatable}{theorem}{thmPureDP}\label{thm:eps_bound}
    Let \mech satisfy $(\eps, 0)$-DP. Let dataset $X \sim \dd = \dd_1 \otimes \ldots \otimes \dd_n$ be drawn from a product distribution. Let $\adv$ be an adversary and let \metric be decomposable into $\ell_1 \cdots \ell_n$. Then, for all $v\in \R$ and for all mechanism outputs $\mechout \subseteq Supp(\mech)$:
    \begin{align*}
        \Pr_{X\sim \dd}\left[\metric(X, \adv(\mechout)) \geq v | \mech(X) = \mechout\right] 
        \leq \Pr_{S_i\sim \bern(\beta_i(\targets, \eps))}\left[\sum_{i\in[n]}S_i \geq v\right],
    \end{align*}
    where $\targets = \adv(\mechout)$ and $\beta_i(\targets, \eps) =  \frac{e^\eps}{e^\eps - 1 + \frac{1}{\Pr_{X\sim \dd_i}[\ell_i(X,\targets)=1]}}$
\end{restatable}

To put this bound into context, consider a perfectly private $(0,0)$-DP mechanism. The bound tells us that, for every fixed mechanism output $a$ and attacker output $\targets$, the number of successfully attacked targets is distributed like a sum of Bernoullis with a flipping probabilities $\Pr_{X\sim \dd_i}[\ell_i(X, \targets) = 1]$ for $i=1, \ldots, n$. As expected, this is exactly the adversary's prior success probability for $\targets$ on a fresh sample of targets from the prior.

Now consider $(\eps, 0)$-DP mechanisms for $\eps > 0$. The bound tells us exactly how far the posterior success rate on each target can be from that of the perfectly private mechanism. Rather than a posterior of $\Pr_{X\sim \dd_i}[\ell_i(X, \targets) = 1]$, which would indicate the adversary gains no information from the mechanism output, the posterior essentially has an $e^\eps$ multiplicative dependence.\footnote{This $e^\eps$ is the dependence used in Reconstruction Robustness, though the actual $\beta$ in Theorem~\ref{thm:eps_bound} is significantly tighter for large priors and large epsilons.} For decomposable success metrics, we can sum the posterior probabilities to get the resulting bound over the total number of successfully attacked targets. 

The dependence on the prior attack success probability captures the phenomenon that for natural (non-uniform) data distributions, some targets are easier to ``attack'' using only prior knowledge, and some success measures are more permissive than others (for example approximate rather than exact reconstruction measures). Furthermore, it captures the tradeoff between the \emph{inherent difficulty} of the attack and the protections guaranteed by DP--something which previously was only stated in vague terms and not quantified clearly.\footnote{For example, when discussing why reconstruction attacks fail against DP mechanisms with large $\eps$s \cite{ponomareva2023dp} states ``There is a natural intuition for why larger $\eps$s provide effective protection in these works--the attacks generally consider an adversary attempting to answer a high-dimensional question (e.g., reconstructing a full training example) with only limited information about the dataset (e.g., distributional).`` We can now formalize this intuition.} In Proposition~\ref{prop:pure_bnd_tight} we show Theorem~\ref{thm:eps_bound} is tight for distributions that are sufficiently close to uniform. The analysis of \citet{keinan2025well} on which types of DP mechanisms yield tight auditing bounds in this ``one run'' framework, could likely be extended to our more general bounds, although the focus of our work is not auditing.

 There are a number of ways to view and use this bound. As in the \emph{privacy auditing} literature \cite{ding2018detecting, jagielski2020auditing, steinke2024privacy, annamalai2025hitchhiker}, one may interpret the bound as a hypothesis test to detect violations of the DP guarantee for a particular mechanism that indicate it was not implemented correctly. If a real attack is more successful than $v$ with high probability then one can confidently refute the hypothesis that the mechanism is differentially private with the claimed privacy parameters. In this work, we are more interested in characterizing the privacy leakage of faithfully implemented DP mechanisms. To that end, we measure $v$ for different probabilities and compare that to the success of empirical attacks.

\subsection{Encoding Attacks}
\label{sec:encoding_attacks}
Our bound is flexible enough to encode most attack settings from the taxonomy developed by Cummings et al.~\cite{cummings2024attaxonomy}.

\paragraph{Attack targets and goals.}
The target distribution $\dd$ represents the adversary's prior over the unknown part of each data record. We can encode MI by setting $X_i \sim \bern(0.5)$ to denote the membership bit of record $i \in [n]$ (as in \citep{steinke2024privacy}). The fixed data records in the MI game are encoded in $\mech$ and are thus known to the adversary.
For full reconstruction, $\dd_i$ becomes the data-generating distribution for person $i$, while AI interpolates these by conditioning $\dd_i$ on known attributes.

\paragraph{Success metrics.}
Our attack success metric captures exact match $\ell_i(X_i, \guesses) = \I\{X_i = \guesses_i\}$ or complex individual-level approximate reconstruction by thresholding distance  $\ell_i(\ds_i, \guesses) = \I\{(\min_j ||\ds_i - \guesses_j||_2) \leq \tau\}$ as done in prior work~\cite{balle2022reconstructing,hayes2024bounding} .

\paragraph{Prior knowledge.}
In Theorem~\ref{thm:eps_bound}, the target generation distribution $\dd_i$ for each person $i\in[n]$ can be set \emph{arbitrarily}, as long as it is independent from others' target distributions. The bound itself only depends on the prior probability of success for a particular attack attempt \guesses: $\Pr_{X\sim \dd_i}[\ell_i(X,\targets)=1]$, that is, the probability that a freshly sampled record $X \sim \dd_i$ would be ``close to'' the attack attempt \guesses.  

\subsection{Bounding $f$-DP}

Generalizing our pure-DP to approximate notions of DP comes with a number of challenges. Approximate DP mechanisms suffer from potentially unbounded privacy loss, so incorporating the mechanism output requires more care. We apply the proof techniques from \citet{steinke2024privacy} and extend their bounds in two key ways: by allowing for general success metrics and arbitrary data distributions. We state and prove approx-DP bounds in Appendix~\ref{sec:approx_df_proofs}, which we further generalize to $f$-DP in Appendix~\Cref{sec:f_dp_proofs} by invoking the primal-dual perspective of $f$-DP~\cite{dong2019gaussian}. 

We begin by stating an $f$-DP generalization of our pure DP bound which incurs an additive factor of $n \cdot \delta_f(\eps)$ in expectation,  where $\delta_f(\eps) = 1 + f^*(-e^\eps)$ and $f^*$ is the convex conjugate of $f$. Here, we trade the simplicity of the bound with the computational overhead of having to estimate the probability over multiple runs of the mechanism.

\begin{restatable}{theorem}{corfBndBasic} \label{cor:f_basic} 
    Let 
    \mech satisfy $f$-DP. Let dataset $X \sim \dd = \dd_1 \otimes \ldots \otimes \dd_n$ be drawn from a product distribution. Let $\adv$ be an attacker and let \metric be decomposable into $\ell_1 \cdots \ell_n$. For all $\eps > 0$, $v\in \R$:
    \begin{align*}
        \Pr_{\substack{X\sim \dd \\ \mechout \sim \mech(X)}}[\metric(X, \adv(\mechout)) \geq v] \leq \Pr_{\substack{X \sim \dd \\ \mechout \sim \mech(X) \\ S_i\sim \bern( \beta_i(\adv(\mechout), \eps))}}\left[\sum_{i \in [n]} S_i \geq v\right] + n \cdot \delta_f(\eps),
    \end{align*}
    where $\beta_i(\adv(\mechout), \eps)= \beta_i(\adv(\mechout), \eps) =  \frac{e^\eps}{e^\eps - 1 + \frac{1}{\Pr_{X\sim \dd_i}[\ell_i(X,\adv(\mechout))=1]}}$. 
\end{restatable}
The approx-DP analog of this bound is \Cref{cor:approx_simple}. Note that since $\delta_f$ is typically much smaller than $1/n$, the additive term brought on by $f$-DP is fairly small, however the flipping probabilities $\beta_1, \ldots, \beta_n$ are now random variables, which must be estimated from multiple runs of the mechanism on fresh data samples. In \Cref{sec:f_dp_proofs} \Cref{thm:eps_bound_f_comp}, we provide a version of this bound which uses the worst-case prior and avoids a dependence on the mechanism output on the RHS. 
\subsection{Numerical Comparisons With Prior Work}
\label{sec:num_exps}

\begin{figure}[h]
    \centering
    \includegraphics[width=0.7\linewidth]{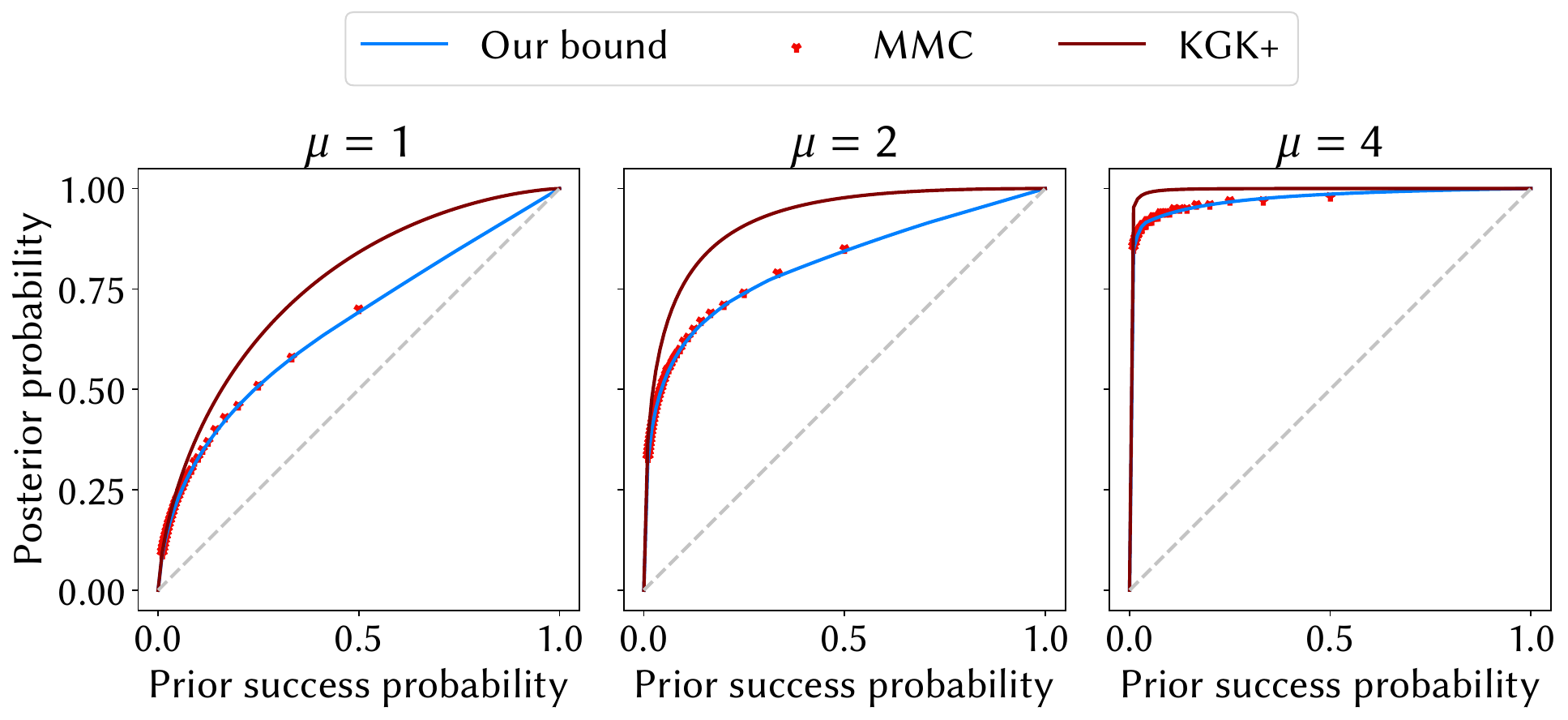}
    \label{fig:compare_bounds_gdp}
    \caption{Compare our bounds with prior state-of-the-art bounds (KGK+~\cite{kulynych2025unifying} and MMC~\cite{mahloujifar2024auditing}) when attacking a single record with varying prior at different privacy levels for $\mu$-GDP. %
    }
    \label{fig:compare_bounds}
\end{figure}

We run numerical experiments to showcase the versatility of our bound and compare it with the state-of-the-art bounds. In order to compare the bounds in a common attack setting, we consider the single target record ($n=1$) and single attempt ($k=1$) setting. To that end, we report our upper bounds on the attack success probability from Theorem~\ref{thm:eps_bound_f_comp} for varying prior probabilities in Figure~\ref{fig:compare_bounds} and compare this to the current state-of-the-art bounds KGK+~\cite{kulynych2025unifying} and MMC~\cite{mahloujifar2024auditing}.

For all privacy levels, our work is substantially tighter than the current state-of-the-art bounds that supports non-uniform priors, KGK+.
For uniform priors, our bounds closely align with the MMC bound.
Unlike MMC, our bounds are continuous and can be extended to any given prior success probability between 0 and 1.
\section{Bounding Data Extraction from LLMs}
\label{sec:novel_bounds}

We now leverage the versatility of our bound to present novel reconstruction bounds for attack settings involving non-uniform priors and multiple target samples.
Specifically, we focus on data extraction attacks in the context of LLMs~\cite{carlini2019secret,carlini2021extracting,lukas2023analyzing,nasr2025scalable}.
Although prior work has shown some success in effectively training LLMs on private data with DP~\citep{li2022large,yu2022differentially,sinha2025vaultgemma}, sensitive data can still be extracted from these models~\cite{lukas2023analyzing}.
The main reason for this is that data can often be inferred from the context present in the sample~\cite{lukas2023analyzing}, similar data appears elsewhere in the dataset~\cite{liu2025language}, or the data itself has low entropy (e.g., 123456789 is a very common password).
In other words, the \emph{prior probability} of extracting data from language models is often non-uniform.
Previous bounds~\cite{hayes2024bounding,mahloujifar2024auditing} either supported single target canaries or uniform priors, but not both simultaneously.

To that end, we focus on three types of data distributions---\textbf{uniform random canary}, \textbf{numerical password}, and \textbf{PII}---and present methods to estimate their prior distributions.
Subsequently, use the Secret Sharer~\citep{carlini2019secret} setup and DP-finetune GPT-2 on canaries drawn from these distributions, try to extract the canaries, and compare  our theoretical bounds to concrete attacks. All our experiments are run on a single A100 GPU. Our experiments on tabular data can be found in Appendix~\ref{sec:exps}. 

\subsection{Estimating Priors}
\paragraph{Uniform random canary.}
Uniformly chosen canaries are most commonly used when evaluating data extraction attacks, as the explicit prior distribution makes it easy to determine when the attack is non-trivial.
Specifically, we derive bounds for the 9-digit uniform random canary from prior work~\cite{carlini2019secret} and set the prior probability to $10^{-9}$ for all possible canaries.

\paragraph{Numerical password.}
Data extraction attacks have also been used to extract sensitive data that appears \emph{naturally} in the dataset such as credit card and social security numbers~\cite{carlini2019secret}.
Along this line, we consider 9-digit numerical passwords as they follow a similar format to uniformly chosen canaries.
Following prior work in similar domains~\cite{gadotti2022pool,erlingsson2014rappor,wang2017locally}, we model the prior distribution with a Zipf's law distribution under the assumption that a few passwords are exceedingly popular amongst users, e.g., the password `123456789' is known to be used roughly 43M times in the wild~\cite{hunt2025pwned}.

\paragraph{PII.}
Lastly, we consider the extraction of Personally Identifiable Information (PII)~\cite{lukas2023analyzing,carlini2021extracting,nasr2025scalable}.
However, unlike uniformly chosen canaries and numerical passwords, PIIs do not have a fixed format or a well-known distribution, and can possibly depend on the surrounding text~\cite{lukas2023analyzing}.

Therefore, we use a pre-trained LLM to estimate the prior over PIIs.
Following recent work~\cite{lukas2023analyzing} we focus on reconstructing names in the Enron emails dataset~\cite{klimt2004introducing} given the surrounding text.
We first sample 500 sequences from the Enron emails dataset and extract a set of candidate PIIs that occur in these sequences.
Next, we calculate the perplexity of the pre-trained GPT-2 model~\cite{radford2019language} for each sequence and candidate PII conditioned on the surrounding text.
Lastly, we normalize the perplexities over the set of candidate PIIs to compute the prior distribution of PIIs given the surrounding text.

\begin{figure}[t]
\centering
    \captionsetup[subfigure]{justification=centering}
    \subfloat[$\textrm{Adv}$ bounds for different prior distributions and privacy levels $\eps$]{
        \includegraphics[width=0.3\linewidth]{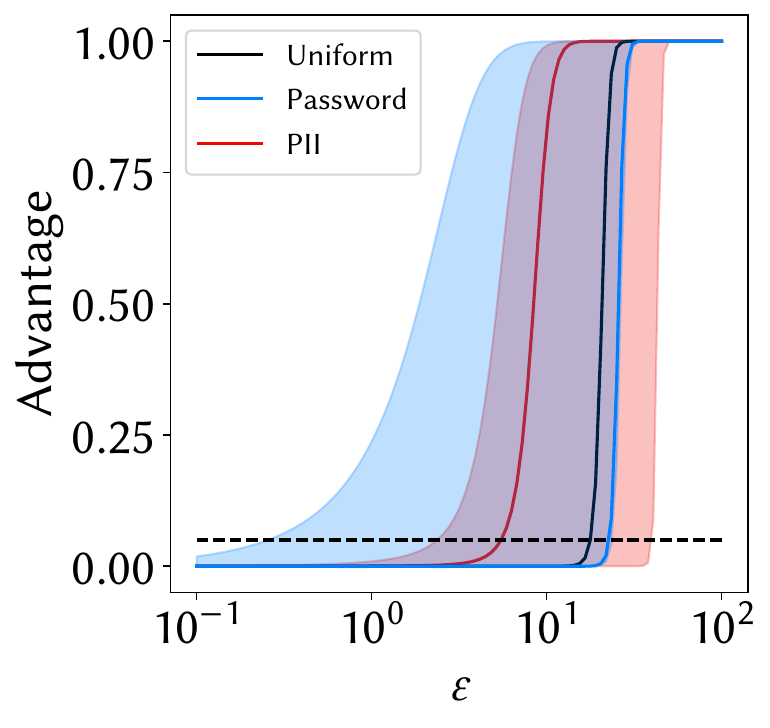}
        \label{fig:novel_bounds}
    }
    \hfil
    \subfloat[$\eps_{\textrm{protect}}$ that bounds $\textrm{Adv} \leq 0.05$ for varying prior probabilities]{
        \includegraphics[width=0.3\linewidth]{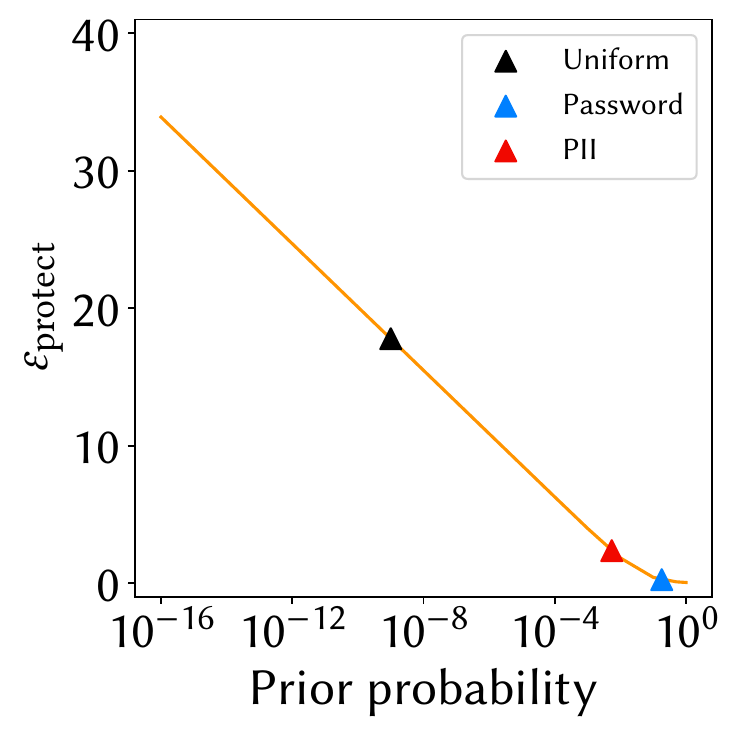}
        \label{fig:eps_protect}
    }
    \caption{Bounding advantage of data extraction for different prior distributions at fixed $\delta = 10^{-5}$. NB: Dotted line represents $\textrm{Adv} = 0.05$ and markers correspond to canary with highest prior probability for each distribution.}
\end{figure}

\subsection{Reconstructing Data from an LLM}
To contextualize our bounds with respect to a concrete attack, we fine-tune a GPT-2 model with DP-SGD~\cite{abadi2016deep} on a dataset consisting of canaries drawn from different distributions and task the adversary with reconstructing the canaries. We emphasize that the aim of this section is not to perform a novel and realistic privacy attack, but rather study how well our bounds model concrete attacks against real mechanisms and therefore take a DP auditing approach~\cite{panda2025privacyauditinglargelanguage}.

Following prior work~\cite{panda2025privacyauditinglargelanguage}, we construct canaries using new tokens, which has been shown to be essential in achieving meaningful attack success against DP fine-tuned LLMs on uniform canaries.
Specifically, each canary consists of two newly created tokens: the prefix, which is fixed to the ID of the canary (\texttt{<1>}, \texttt{<2>}, etc.) and the secret, which is sampled from 10 newly created tokens for each canary.
For simplicity, we scale down the prior distributions from their original domains to being over 10 possible secret tokens.
To reconstruct the canary, we instantiate the ``Model Attack'' on the finetuned model from prior work~\cite{hayes2024bounding}.
We assume that the adversary has access to only the final trained model, the prior distribution over the secret tokens, and the prefix token.
The adversary then reconstructs the secret token by choosing the token with the highest probability assigned by the model conditioned on the prefix, combined with the prior. 

\paragraph{Idealized Gaussian Attack.} While prior "one-run" auditing frameworks have evaluated empirical membership inference attacks against real models~\cite{mahloujifar2024auditing,steinke2024privacy}, extending these audits to data extraction with non-uniform priors introduces a new confounding variable: the practical data extraction attack itself may be highly suboptimal. To isolate the tightness of our theoretical bound from the empirical limitations of the black-box model attack, we also evaluate an idealized Gaussian mechanism. We sample one-hot canary vectors from our priors, add Gaussian noise, and execute a Bayes optimal attack. At a high level, this optimal adversary evaluates every candidate canary by combining the observed noisy vector with the known prior distribution, and guesses the canary with the highest posterior probability (Appendix~\ref{app:idealized_gaussian}). Because the Gaussian mechanism forms the basis of DP-SGD, and its Bayes optimal adversary is mathematically tractable, it provides a ceiling for attack success. Comparing our bound directly against this Bayes optimal adversary tells us exactly how much slack exists in our bound, independent of the empirical limitations of the Model Attack.

\subsection{Results}
In Figure~\ref{fig:novel_bounds}, we plot our bounds on the generalized advantage ($\textrm{Adv}$)~\cite{cherubin2017bayes} when reconstructing different secrets from the three prior distributions considered.
Since the advantage is a measure of an attack's \emph{significance} and not just its success, it allows us to compare the bounds more effectively between canaries with different priors. Using Theorem~\ref{thm:eps_bound_approx_comp}, we compute posterior bounds for canaries with minimum\footnote{When estimating the prior distribution for PIIs, we note that the minimum probability can sometimes be prohibitively small ($\approx 0$) and hence we plot the $10^\textrm{th}$ percentile instead.}, maximum, and median prior probabilities. We convert these to generalized advantage bounds using the following equation~\cite{cherubin2017bayes}: $\textrm{Adv} = (\textrm{Posterior} - \textrm{Prior})/(1 - \textrm{Prior})$.

Additionally, in Figure~\ref{fig:eps_protect} we vary the prior probability and plot the corresponding $\eps_{\textrm{protect}}$ that would result in an $\textrm{Adv} \leq 0.05$ at fixed $\delta = 10^{-5}$.
Informally, Figure~\ref{fig:eps_protect} plots the privacy level required to protect data with a given prior probability from leaking.\footnote{Here, $0.05$ is chosen arbitrarily, but depending on the context the advantage threshold considered can be made more stringent or loose.}

\paragraph{Non-Uniform Priors Drastically Increase Risk.}

While it is intuitive that common secrets are easier to guess, our framework explicitly quantifies how non-uniform priors result in disparate privacy leakage across individuals. Unlike uniform canaries, which share identical advantage bounds, canaries with higher prior probabilities admit attacks with substantially higher adversarial advantage. For example, at $\eps=1$, the advantage bound for extracting the median numerical password is $\ll0.05$, but jumps to $0.24$ for the most probable password. Similarly, at $\eps=10$, the maximum advantage for any uniform canary is negligible, whereas the most probable PII and password canaries reach an advantage bound of $\approx1.0$. Ultimately, evaluating mechanisms solely on uniform canaries drastically underestimates the theoretical privacy leakage for individuals holding highly probable secrets.

\paragraph{Calibrating Epsilon to Reconstruction Risks.}
A natural research question that arises from our results is what would be an appropriate privacy level that protects against the different types of data extraction.
To answer this, for each prior distribution, we compute the $\eps_\textrm{protect}$ value that bounds the advantage for the worst-case canary to $\leq 0.05$ and plot it in Figure~\ref{fig:eps_protect}.
Subsequently, we note that in general as canaries become less probable, they become easier to protect even at low privacy levels e.g., a uniformly drawn 9 digit canary can be protected by $\eps_\textrm{protect} = 17.8$ but numerical passwords and PIIs require $\eps_\textrm{protect} = 0.25$ and $\eps_\textrm{protect} = 2.37$, respectively.
This shows that, even though a large $\eps$ has in the past proven effective at preventing the leakage of uniform canaries~\cite{carlini2019secret}, naturally occurring canaries might require a substantially smaller $\eps$ to be protected.
Furthermore, this also illustrates how practitioners can use our bounds in practice to calibrate the privacy parameters to specific privacy risks.

\begin{figure}[t]
    \centering
    \includegraphics[width=0.8\linewidth]{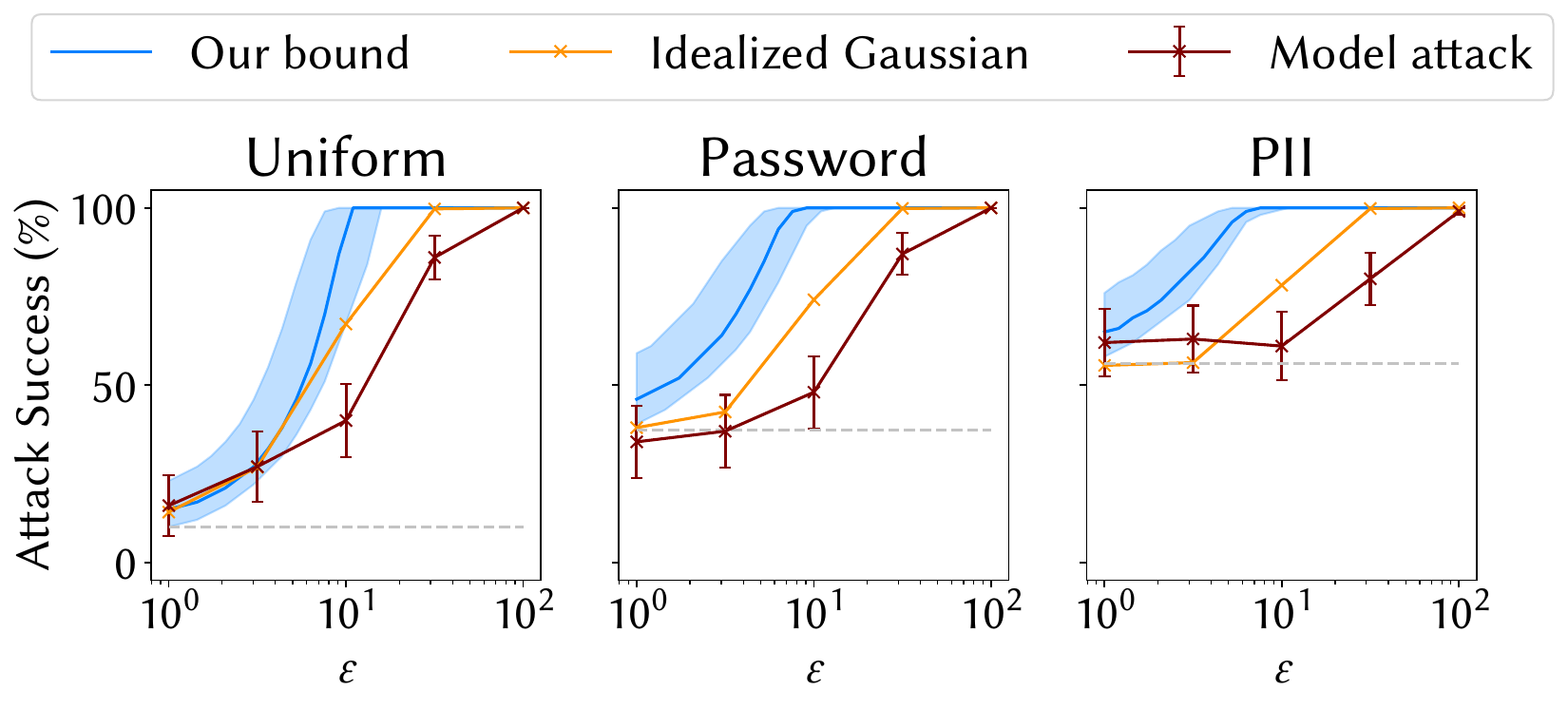}
    \caption{Compare our bounds with success of reconstructing 100 canaries (with 95\% CI) from a DP fine-tuned GPT-2 model at different privacy levels and fixed $\delta = 10^{-5}$. NB: Dotted gray line represents baseline attack w/o access to the DP mechanism.}
    \label{fig:llm}
\end{figure}

\paragraph{The Gap Between Theoretical and Concrete Attacks.}
In Figure~\ref{fig:llm} we compare our theoretical bounds (Theorem~\ref{thm:eps_bound_f_comp}) against concrete reconstruction accuracies of the model attack on GPT-2 fine-tuned on a 100 canaries. We perform hyper-parameter tuning to identify the best clipping norm $[1.0, 10.0, 100.0]$, learning rate $[0.01, 0.1, 1, 10]$, and number of epochs $[1, 3, 10, 32, 100]$ for each privacy level studied. By overlaying the Idealized Gaussian attack, Figure~\ref{fig:llm} clearly isolates two distinct gaps between our theoretical bound and the concrete model attack.

First, comparing our bound to the Bayes optimal baseline reveals the mathematical slack in our framework. For uniform canaries, our bounds align almost perfectly with the idealized attack. However, for non-uniform priors (passwords and PII), a gap persists between the bound and the optimal attack. This indicates that while our bounds are tight for uniform distributions, there is still room for theoretical refinement to tighten the bounds for mechanisms handling highly skewed priors.\footnote{It's not clear we could achieve a tighter bound without losing the generality of our framework, but perhaps one could derive a tighter bespoke bound for the Gaussian mechanism.}

Second, comparing the Idealized Gaussian to the Model Attack reveals that the model attack falls significantly short of the Bayes optimal baseline, particularly in the medium-privacy regime ($\eps=10$) and for non-uniform priors. On the other hand, it is close to the Idealized Gaussian for high privacy ($\varepsilon = 1.0$) and vacuous regimes ($\varepsilon = 100.0$). This empirical gap likely stems from documented limitations of black box auditing of DP-SGD, like the fact that the bound and idealized attack capture total privacy loss across all iterations, but the black-box model attack only observes the final model weights~\cite{annamalai2024nearly,cebere2025tighter,nasr2025the}. The attack would miss data that the model memorized early on and ``forgot'' during later epochs. 

\section{Conclusion}
\label{sec:conclusion}
In this paper, we introduced a novel framework that derives precise privacy leakage bounds for DP mechanisms under previously unexplored, realistic threat models—specifically those involving multiple targets, non-uniform prior distributions, and complex reconstruction metrics. Our framework equips practitioners with the theoretical foundation needed to interpret DP parameters and calibrate them against context-specific privacy risks rather than relying on worst-case heuristics.

\paragraph{Limitations and Future Work.}
While our bounds currently rely on decomposable success metrics, developing bounds for, non-decomposable success metrics remains an open challenge that would capture an even broader set of adversarial goals. Additionally, we encourage the privacy community to design and study more empirical DP attacks that explicitly incorporate non-uniform priors, as our results show these distributions dramatically alter the true risk landscape and can be harder to execute. Finally, adversaries rarely possess perfect knowledge of the true data-generating distribution. An important next step is understanding how these privacy bounds behave when the adversary only has access to an approximation of the exact prior.

\paragraph{Broader Impacts.} 
Our framework equips practitioners to calibrate privacy parameters to context-specific threat models. The primary risk of this framework is the potential for false confidence or policy misuse. Because our bounds can demonstrate that empirical risks are lower than worst-case MI guarantees, practitioners might be tempted to  justify weaker privacy protections. These bounds rely on strict structural assumptions (which may not hold for all use cases) and should not be used as a blanket reason to degrade privacy safeguards.

\begin{ack}
This work has been supported by the National Science Scholarship (PhD) from the Agency for Science, Technology and Research, Singapore. The authors would also like to thank Bogdan Kulynych for helpful discussions.
\end{ack}

{\small
\bibliographystyle{abbrvnat}
\bibliography{arxiv_main}
}

\appendix
\section{Pure DP Proofs} \label{sec:deferred_proofs}

Recall the definition of stochastic dominance.

\begin{definition}[Stochastic Dominance] Let $X, Y \in \R$ be random variables. We say $X$ is \emph{stochastically dominated} by $Y$ if for all $t \in \R$, $\Pr(X \geq t) \leq \Pr(Y \geq t)$. 
\end{definition}

\thmPureDP*
\begin{proof}
   Fix a product distribution distribution $\dd = \dd_1 \otimes \cdots \otimes \dd_n$, a $(\eps, 0)$-DP mechanism \mech, a metric \metric that decomposes into $\ell_1\cdots \ell_n$, and an attacker $\adv$, and fix any mechanism output $\mechout$. We can assume WLOG that the attacker $\adv$ is deterministic. Let $\targets = \adv(\mechout)$ be the attack attempts for our fixed mechanism output $\mechout$. 
   
   Following the techniques of Steinke, Nasr, Jagielski, we will prove the upper bound by induction over data records. Fix some $i \in [n]$ and fix $x_{< i} = (x_1, \ldots, x_{i-1}) \in Supp(\dd_1 \otimes \cdots \otimes \dd_{i-1})$. Then, applying Bayes' theorem and the $\eps$-DP guarantee:
    \begin{align*}
       & \Pr[\ell_i(X_i)  = 1 | \mech(X) = \mechout, X_{< i} = x_{<i}]\\
       & = \frac{\Pr[\mech(X) = \mechout | \ell_i(X_i) = 1, X_{< i} = x_{<i}] \cdot \Pr[\ell_i(X_i) = 1 | X_{< i} = x_{<i}]}{\Pr[\mech(X) = \mechout | X_{< i} = x_{<i}]}\\
       & = \frac{\Pr[\mech(X) = \mechout | \ell_i(X_i) = 1, X_{< i} = x_{<i}] \cdot \Pr[\ell_i(X_i) = 1]}{\sum_{b \in \{0,1\}}\Pr[\mech(X) = \mechout | \ell(X_i, \targets) = b, X_{< i} = x_{<i}] \cdot \Pr[\ell_i(X_i) = b]}\\
       & = \frac{1}{1 + \frac{\Pr[\mech(X) = \mechout | \ell_i(X_i) = 0, X_{< i} = x_{<i}] \cdot \Pr[\ell_i(X_i)=0]}{\Pr[\mech(X) = \mechout | \ell_i(X_i)=1, X_{< i} = x_{<i}] \cdot \Pr[\ell(X_i,\targets)=1]}}\\
       & \leq \frac{1}{1 + e^{-\eps} \cdot \frac{\Pr[\ell_i(X_i)=0]}{\Pr[\ell_i(X_i)=1]}}\\
       & = \frac{1}{1 + e^{-\eps} \cdot \frac{1-\Pr[\ell_i(X_i)=1]}{\Pr[\ell_i(X_i)=1]}}\\
       & = \frac{e^\eps}{e^\eps - 1 + \frac{1}{\Pr[\ell_i(X_i)=1]}}.
    \end{align*}
    
    By a similar argument,
    \begin{align*}
       & \Pr[\ell_i(X_i) = 1| \mech(X) = \mechout, X_{< i} = x_{<i}]\\ 
       & \in \left[\frac{e^{-\eps}}{e^{-\eps} - 1 + \frac{1}{\Pr[\ell_i(X_i)=1]}}, ~~  \frac{e^\eps}{e^\eps - 1 + \frac{1}{\Pr[\ell_i(X_i)=1]}}\right].
    \end{align*}

    Thus,
    \begin{align*}
        & \Pr[\ell_i(X_i, \targets) = 1| \mech(X) = \mechout, X_{< i} = x_{<i}]\\
        & \leq \frac{e^\eps}{e^\eps - 1 + \frac{1}{\Pr[\ell_i(X_i)=1]}} := \beta_i(\targets, \eps),
    \end{align*}

    Note that $\beta_i(\targets, \eps) \in [0,1]$ for all $\eps$ and all $\targets \in Supp(\dd)$. Now we will proceed with induction. Let $W_{i-1} = \sum_{m\in [i-1]} \ell_m(X_m, \targets)$ be the loss metric the first $i-1$ data records. Assume inductively that $W_{i-1}$ is stochastically dominated by $\hat{W}_{i-1} = \sum_{m\in [i-1]} S_m$ where $S_m \sim \bern(\beta_i(\targets, \eps))$. We have that, conditioned on $W_{i-1}$ the variable $\ell_i(X_i, \targets)$ is stochastically dominated by $\bern(\beta_i(\targets, \eps))$. Thus, we may apply Lemma 4.9 from Steinke, Nasr, Jagielski to conclude that $W_i = W_{i-1} + \ell_i(X_i, \targets)$ is stochastically dominated by $\hat{W}_i := \sum_{i\in [n]} S_i$ where $S_i \sim \bern(\beta_i(\targets, \eps))$.
\end{proof}

\subsection{Our bound is tight for RR with priors that are sufficiently close to uniform}
\label{sec:tight_rr}
Next, we show that the bound in Theorem~\ref{thm:eps_bound} is tight for randomized response on distributions that are not too far from uniform. 

\begin{proposition}\label{prop:pure_bnd_tight} 
    For some domain $[m]$ let $RR_\eps: \cX^n \to [m]$ be the mechanism that applies randomized response independently to each data record, and let \metric be the zero-one reconstruction loss. Then, for every data distribution \dd such that $\dd(a) \leq e^\eps \cdot \dd(b)$ for all $a, b\in [m]$, the bound in Theorem~\ref{thm:eps_bound} is achieved by the Bayes' optimal reconstruction attacker. That is, the bound is tight for this mechanism.
\end{proposition}

\begin{proof}
    Fix $n = k > 0$, and fix data domain $[m]$. Let $\metric(\ds, \guesses) = \sum_{i\in [n]} \I\{\ds_i = \guesses_i\}$. The optimal prior-aware adversary $R$ against randomized response will operate on each term $a \in [m]$ independently and for each one output $$ R(a, \dd, \eps) = argmax_{b\in [m]} \Pr_{X\sim \dd}(X=b| RR_\eps(X) = a).
    $$
    Thus, we have 
    \begin{align*}
        & \Pr_{X\sim \dd}(X=b| RR_\eps(X) = a)\\ 
        & = \frac{\Pr(RR_\eps(X) = a | X = b) \cdot \Pr(X=b)}{\Pr(RR_\eps(X) = a)}.
    \end{align*}
    Note that the conditional probability in the numerator is
    \begin{align*}
    & \Pr(RR_\eps(X) = a | X = b)\\
    & = \I\{a = b\} \cdot \frac{e^\eps}{e^\eps -1 + m} + \I\{a \neq b\} \cdot \frac{1}{e^\eps -1 + m}
    \end{align*}
    and the denominator is
    \begin{align*}
    & \Pr(RR_\eps(X) = a)\\
    & = \dd(a) \cdot \frac{e^\eps}{e^\eps -1 + m} + (1-\dd(a)) \cdot \frac{1}{e^\eps -1 + m}.
    \end{align*}
    
    Thus, when we simplify the terms we get
    \begin{align*}
        & \Pr_{X\sim \dd}(X=b| RR_\eps(X) = a)\\
        & = \begin{cases}
            \frac{\dd(a) e^\eps}{\dd(a) (e^\eps - 1) + 1} & \text{ if } a=b\\
            \frac{\dd(b)}{\dd(a) (e^\eps - 1) + 1} & \text{ if } a\neq b.\\
        \end{cases}
    \end{align*}
    
    Thus, the optimal adversary will use the following strategy:
    $$ R(a, \dd, \eps) = 
    \begin{cases}
        a & \text{ if } \dd(a) \cdot e^\eps \geq \dd(b) \text{ for all } b \in [m]\setminus \{a\}\\
        b^* = argmax_{b \in [m]} \dd(b) & \text{ otherwise}
    \end{cases}, $$
    that is, $R$ \emph{ignores} the prior if the mechanism output is consistent with the prior and otherwise \emph{defaults} to the prior.

    Now that we have defined the optimal adversary, for every fixed $a \in [m]$ we can show that the bound is tight for distributions $\dd$ such that $\dd(a) \cdot e^\eps \geq \dd(b)$ for all $b\in [m]$. 

    In this case, optimal attacker will output $a$ as its reconstruction guess. So, the attacker's success rate is simply 
    \begin{align*}
    & \Pr(X = a|RR_\eps(X) = a)\\
    & = \frac{\dd(a) e^\eps}{\dd(a) (e^\eps - 1) + 1} = \frac{e^\eps}{e^\eps - 1 + \frac{1}{\dd(a)}},
    \end{align*}
    which exactly matches the upper bound in Theorem~\ref{thm:eps_bound}.

\end{proof} %
\section{Approximate DP Theorems and  Proofs}\label{sec:approx_df_proofs}

\begin{restatable}[Extension of Prop 5.7 in \cite{steinke2024privacy}]{theorem}{thmApproxBndNonComp} \label{thm:approx_non_comp} 
    Let \mech satisfy $(\eps, \delta)$-DP. Let dataset $X \sim \dd = \dd_1 \otimes \ldots \otimes \dd_n$ be drawn from a product distribution. Let $\adv$ be an attacker and let \metric be decomposable into $\ell_1 \cdots \ell_n$. Then, for all $v\in \R$ and for all mechanism outputs $\mechout \subseteq Supp(\mech)$:
    
    \begin{align*}
        \Pr_{X\sim \dd}[\metric(X, \adv(\mechout)) \geq v |\mech(X) = a ] \leq \Pr_{S_i\sim \bern( \beta_i(\targets, \eps))}\left[ G(\mechout) + \sum_{i \in [n]} S_i \geq v\right],
    \end{align*}
    where $\targets = \adv(\mechout)$ and $\beta_i(\targets, \eps) =  \frac{e^\eps}{e^\eps - 1 + \frac{1}{\Pr_{X_i\sim \dd_i}[\ell_i(X_i,\targets)=1]}}$ and where $G: Range(\mech)^n \to \{0,\ldots, n\}$ is independent from $S_1, \ldots, S_n$ with $\E_{\Mechout \sim \mech(X),G}[G(\Mechout)] = n\cdot \delta$. 
\end{restatable}

\begin{proof}
    The proof follows same general proof structure as that of Proposition 5.7 in \citet{steinke2024privacy}.

    Fix a product distribution $\dd = \dd_1 \otimes \ldots \dd_n$, an $(\eps, \delta)$-DP mechanism \mech, a mechanism output $\mechout$, and an attack $\adv$. For $i \in [n] \cup \{0\}$, $x_{\leq i} \in \cX^i$, and $b \in \{0,1\}$, let $\mech(x_{<i})$ denote the distribution of $\mech(X)$ where $X\sim \dd$ conditioned on $X_{<i} = x_{<i}$. In an abuse of notation, we will use the shorthand $\mech(x_{<i}, b)$ to denote the distribution of $\mech(X)$  for $X\sim \dd$ conditioned on $X_{<i} = x_{<i}$ and $\ell_i(X_i, \adv(\mechout)) = b$ for $b\in \{0,1\}$. 
    
    Now, for distributions $P$ and $Q$ on $Range(\mech)$, let $E_{P,Q}$ be the randomized function promised by Lemma 5.6 in \cite{steinke2024privacy}. As in their analysis, the internal randomness of $E_{P,Q}$ is independent of everything else. Given a fixed sequence $x_1,...,x_n$, we will apply the lemma once for each index $i=1,...,n$, with $P = \mech(x_{<i}, 0)$ and $Q = \mech(x_{<i}, 1)$. 
    
    To simplify notation, in the remainder of the proof, we assume that the adversary's prediction algorithm $\adv$ is simply the identity; we can think of $\adv$ as being folded into  the loss functions $\ell_i$. 
    Thus we set $\targets = \adv(\mechout) = \mechout$ and aim to bound $\metric(X,\mechout)= \sum_i \ell_i(X_i,a)$.
    
    Thus, for all $i\in [n]$ and for all $x_{<i}$, and all mechanism outputs $\mechout$, we have
    \begin{align}
        \Pr_{\substack{X\sim \dd \\ \mechout \sim \mech(X), E}}
        \left[
            \ell_i(X_i, \mechout) = 1 \land E_{\mech(x_{<i}, 0), \mech(x_{<i}, 1)}(\mechout) = 1 {\Big |} X_{<i} = x_{<i}, \mech(X) = \mechout
        \right] \nonumber 
        \\
        \leq \frac{e^\eps}{e^\eps - 1 + \frac{1}{\Pr(\ell_i(X_i, \mechout) = 1)}}
        = \beta_i(\mechout,\eps),
        \label{eq:betabound}
    \end{align} 
    and for $b \in \{0,1\}$,
     \begin{align}
        \Pr_{\substack{X\sim \dd \\ \mechout\sim \mech(X), E}}\left[
            E_{\mech(x_{<i}, 0), \mech(x_{<i}, 1)}(\mechout) = 1 
            {\Big |} X_{<i} = x_{<i}, \ell_i(X_i, \mechout)=b
            \right] 
            \geq 1-\delta.
    \end{align}

    In the remainder of the proof we will define three quantities: $W$, which is what we want to upper bound but which isn't always well-behaved, $\widetilde{W}$ which is well-behaved but hard to quantify, and $\widehat{W}$ which is well-behaved and easy to quantify and upper-bounds the other two quantities (with an extra term $F$ to handle the poorly-behaved situations).
    
    Now, for fixed $x \in \cX^n$, $k \in [n]$, $\mechout \in Range(\mech)$, 
    let 
    $$\widetilde{W}_{i-1}(x, a) = \sum_{m \in [i-1]} \ell_m(x_m, \mechout) \cdot E_{\mech(x_{<i}, 0), \mech(x_{<i}, 1)}(\mechout)$$ 
    and let 
    $$\widehat{W}_{i-1}(\mechout) = \sum_{m\in [i-1]} S_m \, ,$$ 
    where $S_m \sim \bern(\beta_i(\mechout, \eps))$ independently for each $m \in [i-1]$

    \Cref{eq:betabound} shows that each term $S_m$ (in the $\hat W$'s) stochastically dominates the  term $\ell_m(x_m, \mechout) \cdot E_{\mech(x_{<i}, 0), \mech(x_{<i}, 1)}(\mechout)$
    (in the $\tilde W$'s), even conditioned on the corresponding terms with smaller indices.
    By induction (e.g., see Lemma 4.9 in \citet{steinke2024privacy}), for every $i\in [n]$ and every $\mechout$, the conditionally-distributed random variable $(\widetilde{W}_i(X, \mechout)|\mech(X) = \mechout)$, where $X\sim \dd$,  is stochastically dominated by $\widehat{W}_i(\mechout)$.

    Now, for fixed $\ds\in \cX^n$ and $\mechout \in Range(\mech)$, define

    $$
    G(\ds, \mechout) = \sum_{i\in[n]} \I\{E_{\mech(x_{<i}, 0), \mech(x_{<i}, 1)}(\mechout) = 0\},
    $$
    so that 
    \begin{align*}
    & W_n(\ds, \mechout) := \metric(x, \mechout) := \sum_{i \in [n]} \ell_i(\ds_i, \mechout)\\
    & \leq \widetilde{W}_n(\ds, \mechout) + G(\ds, \mechout).
    \end{align*}

    Since that the conditional distribution $(\widetilde{W}_n(X, \mechout)|\mech(X) = \mechout)$ where $X \sim \dd$ is stochastically dominated by $\widehat{W}_n(\mechout)$, we know that $W_n$ is stochastically dominated by the convolution $\widehat{W}_n(\Mechout) + G(X,\Mechout)$ for $\Mechout \sim \mech(X)$.

    To summarize: 
    \begin{align*}
    &(\widetilde{W}_n(X, a)|\mech(X) = \mechout) \sd \widehat{W}_n(\mechout) \quad \forall \mechout \\
    \implies & W_n(X, \Mechout) \sd \widehat{W}_n(\Mechout) + G(X,\Mechout) 
    \quad \text{for } \Mechout \sim \mech(X) \nonumber
    \end{align*}
    
    Furthermore, $\E[G(\ds, \mechout)] = \sum_{i\in[n]} \Pr[E_{\mech(x_{<i}, 0), \mech(x_{<i}, 1)}(\mechout) = 0] \leq n\cdot \delta.$ Lastly, since $\widehat{W}_{i-1}(\mech)$ does not depend on $X$, the input dataset $X$ does not contribute to the dependence between $G(X,\Mechout)$ and $\widehat{W}_{i-1}(\Mechout)$ so we can ignore this input to $F$ and assume $G(\Mechout) = G(X,\Mechout)$ for $X\sim \dd$.
\end{proof}

\begin{restatable}{theorem}{corApproxBndNonComp} \label{cor:approx_simple} 
    Let 
    \mech satisfy $(\eps, \delta)$-DP. Let dataset $X \sim \dd = \dd_1 \otimes \ldots \otimes \dd_n$ be drawn from a product distribution. Let $\adv$ be an attacker and let \metric be decomposable into $\ell_1 \cdots \ell_n$. Then, for all $v\in \R$:
    \begin{align*}
        \Pr_{\substack{X\sim \dd \\ \mechout \sim \mech(X)}}[\metric(X, \adv(\mechout)) \geq v] \leq \Pr_{\substack{X \sim \dd \\ \mechout \sim \mech(X) \\ S_i\sim \bern( \beta_i(\adv(\mechout), \eps))}}\left[\sum_{i \in [n]} S_i \geq v\right] + n \cdot \delta,
    \end{align*}
    where $\beta_i(\adv(\mechout), \eps) =  \frac{e^\eps}{e^\eps - 1 + \frac{1}{\Pr_{X\sim \dd_i}[\ell_i(X,\adv(\mechout))=1]}}$.
\end{restatable}

Note that since $\delta$ is typically much smaller than $1/n$, the additive term brought on by approximate DP is fairly small, however the flipping probability $\beta$ is now a random variable itself, which must be estimated from multiple runs of the mechanism on fresh data samples.

Next we present an approximate DP bound that gives a tighter dependence on $\delta$ and which can be computed in only one run of the mechanism. The flipping probabilities $\beta$ are computed using the Bayes optimal \textit{a priori} attack $\guesses^*$, which in general will roughly be the $k$ heaviest elements in the domains of $\dd_1, \ldots, \dd_n$.\footnote{If each $\dd_i$ is different and $k < n$ then computing $\guesses^*$ becomes more complicated.} When the prior distributions are close to uniform, then using $\guesses^*$ to compute the prior should be fairly tight, but for highly skewed data distributions, \Cref{thm:eps_bound_approx_comp} will likely be much tighter. 

\begin{restatable}[Generalization of Theorem 5.2 \cite{steinke2024privacy}]{theorem}{thmApproxComp}\label{thm:eps_bound_approx_comp}
Let \mech satisfy $(\eps, \delta)$-DP. Let dataset $X \sim \dd = \dd_1 \otimes \ldots \otimes \dd_n$ be drawn from a product distribution. Let $\adv$ be an attacker and let \metric be decomposable into $\ell_1 \cdots \ell_n$. Then, for all $v\in \R$,
\begin{align*}
    \Pr_{\substack{X \sim \dd \\ \mechout \sim \mech(X)}} [\metric(X, \adv(\mechout)) \geq v] \leq \Pr_{S^*_i \sim \bern(\beta_i(\eps))}\big[\sum_{i\in [n]} S^*_i \geq v\big] + \alpha \cdot n \cdot \delta, 
\end{align*}
where $\beta_i(\eps) = \frac{e^\eps}{e^\eps - 1 + \frac{1}{\Pr_{X\sim \dd_i}[\ell_i(X, \guesses^*) = 1]}}$ where $\guesses^*$ is the \textit{a priori} Bayes optimal attack attempt on \dd and  
\begin{align*}
    \alpha = \max\left\{\frac{1}{j} \left(\Pr_{S^*_i \sim \bern(\beta_i(\eps))}\left[\sum_{i\in [n]} S^*_i \geq v - j\right] -  \Pr_{S^*_i \sim \bern(\beta_i(\eps))}\left[\sum_{i\in [n]} S^*_i \geq v\right] \right) : j \in [n] \right\}.
\end{align*}
\end{restatable}

\begin{proof}

Fix an $(\eps,\delta)$-DP mechanism \mech, fix a metric \metric that is decomposable into $\ell_1 \cdots \ell_n$, and fix an attacker $\adv$. Let $\dd = \dd_1 \otimes \ldots \otimes \dd_n$ be the data distribution.

    We begin by using Theorem~\ref{thm:approx_non_comp} and taking the expectation over the potential mechanism outputs to get
    
    \begin{align*}
        & \Pr_{X\sim \dd, \mechout\sim \mech(X)}[\metric(X, \mechout) \geq v ]\\
        & \leq \Pr_{\substack{\mechout \sim \mech(X) \\ S_i\sim \bern( \beta_i(\mechout, \eps))}}\left[ F(\mechout) + \sum_{i \in [n]} S_i \geq v\right],
    \end{align*}
     where $\beta_i(\mechout, \eps) =  \frac{e^\eps}{e^\eps - 1 + \frac{1}{\Pr_{X\sim \dd_i}[\ell_i(X,\mechout)=1]}}$ and where $F: Range(\mech) \to \{0,\ldots, n\}$ is independent from $S_1, \ldots, S_n$ with $\E_{\Mechout \sim \mech(X),F}[F(\Mechout)] = n\cdot \delta$. 
    
    Since the constraints on $F$ are fairly simple, we will set up a linear program to find the optimal distribution for random variable $F(a)$ for a given mechanism output $a$ and threshold $v \in \R$. In this linear program, the variables are the values on $\Pr(F(\mechout) = j)$ for each $j \in \{0, \ldots, n\}$ and the two constraints are: (1) it must be a probability distribution, and (2) the expectation of $F$ is at most $n \cdot \delta$. So, for a fixed mechanism output $\mechout$ and threshold $v$:

    \begin{alignat*}{2}
        &\text{maximize}  &&  \Pr_{ S_i\sim \bern( \beta_i(\mechout, \eps))}\left[ F(\mechout) + \sum_{i \in [n]} S_i \geq v\right] \\
        & &&= \sum_{j=0}^n \Pr_F[F(\mechout) = j]) \cdot \Pr_{S_1 \ldots S_n}\left[\sum_{i \in [n]} S_i \geq v - j\right] \\
        & \text{subject to} \quad && \E_{F}[F(\mechout)] = \sum_{j = 0}^n \Pr[F(\mechout) = j] \cdot j \leq n \cdot \delta,\\
        & && \sum_{j = 0}^n \Pr[F(\mechout) = j] = 1, \text{and} \\
        & && \Pr[F(\mechout) = j] \geq 0 \quad \forall j\in\{0, \ldots, n\}.
    \end{alignat*}

    By strong duality, the linear program above has the same value as its dual (again, for fixed $\mechout, v$):

    \begin{alignat*}{2}
        &\text{minimize}  &&  n\delta\alpha + \gamma \\
        & \text{subject to} \quad && 
        \alpha \cdot j + \gamma \geq \Pr_{S_1 \ldots S_n}\left[\sum_{i \in [n]} S_i \geq v - j\right]\\
        & && \forall j\in \{0, \ldots, n\},\\
        & && \alpha \geq 0
    \end{alignat*}

    Any feasible solution to the dual gives an upper bound on the primal. So, we can use the solution given by 
    \begin{align*}
        & \gamma = \Pr[S^* \geq v] \quad \text{and}\\
        & \alpha = \max\left(\{0\} \cup \frac{1}{j}\Pr[S^* \geq v - j] - \gamma : j \in [n] \right),
    \end{align*}
    where $S^*$ is any distribution that satisfies \begin{equation*}
        \Pr_{S^*}[S^* \geq v - j] \geq \Pr_{S_i\sim \bern( \beta_i(\mechout, \eps))}\left[\sum_{i \in [n]} S_i \geq v-j\right]
    \end{equation*}
     for all $j \in \{0, \ldots, n\}$ and for all $\mechout \in Supp(\mech).$
     
     In particular, we can take $S^* = \sum_{i \in [n] }S^*_i$ where $S^*_i \sim \bern(\beta_i(\targets^*, \eps)$ where $\targets^*$ maximizes the attack success on $\dd$. 
\end{proof}

\section{Omitted $f$-DP Theorems and Proofs} \label{sec:f_dp_proofs}

First, recall the definition of $f$-DP.

\begin{definition}[Trade-off function~\cite{dong2019gaussian}]
    For any two probability distributions, $P$ and $Q$ on the same space, the trade-off function $T(P, Q): [0, 1] \rightarrow [0, 1]$ is defined as:
    \begin{equation*}
        T(P, Q)(\alpha) \triangleq \inf_{\phi} \{\beta_\phi: \alpha_\phi \leq \alpha\}
    \end{equation*}
    where the infimum is taken over all rejection rules $\phi$ for which $\alpha_\phi$ and $\beta_\phi$ are the type I and type II errors, respectively.
\end{definition}

\begin{definition}[$f$-DP~\cite{dong2019gaussian}]
    \label{def:fdp}
    A randomized mechanism $\mech : \mathcal{D} \rightarrow \mathcal{R}$ satisfies $f$-DP if, for any two adjacent datasets $\ds, \ds' \in \mathcal{D}$, and $\alpha \in [0, 1]$ it holds:
    \begin{equation*}
        T(\mech(\ds), \mech(\ds'))(\alpha) \geq f(\alpha)
    \end{equation*}
\end{definition}

We present two $f$-DP theorems that were omitted from the main body. Then, we will present the proofs of all three $f$-DP bounds.

\begin{restatable}[Generalization of Proposition 5.7~\cite{steinke2024privacy}]{theorem}{thmfDPBndNonComp} \label{thm:f_dp_non_comp} 
    Let \mech satisfy $f$-DP. Let dataset $X \sim \dd = \dd_1 \otimes \ldots \otimes \dd_n$ be drawn from a product distribution. Let $\adv$ be an attacker and let \metric be decomposable into $\ell_1 \cdots \ell_n$. Then, for all $v\in \R$, $\eps > 0$, and mechanism outputs $\mechout \subseteq Supp(\mech)$:
    
    \begin{align*}
        \Pr_{X\sim \dd}[\metric(X, \adv(\mechout)) \geq v |\mech(X) = a ] \leq \Pr_{S_i\sim \bern( \beta_i(\targets, \eps))}\left[ G(\mechout) + \sum_{i \in [n]} S_i \geq v\right],
    \end{align*}
    where $\targets = \adv(\mechout)$ and $\beta_i(\targets, \eps) =  \frac{e^\eps}{e^\eps - 1 + \frac{1}{\Pr_{X_i\sim \dd_i}[\ell_i(X_i,\targets)=1]}}$ and where $G: Range(\mech)^n \to \{0,\ldots, n\}$ is independent from $S_1, \ldots, S_n$ with $\E_{\mechout \sim \mech(X),G}[G(\mechout)] = n\cdot \delta_f(\eps)$. 
\end{restatable}

The main challenge with this bound is that $G$ is unspecified and generally depends on the mechanism.

In order to avoid flipping probabilities that depend on the attack as in \Cref{cor:f_basic} (and thus require multiple runs of the mechanism to estimate), we present a version of this bound that uses the worst-case prior instead. We compute
the $\beta_1, \ldots, \beta_n$ with respect to the guesses $\guesses^*$ that maximize the attack success on \dd. This approach provides a single distribution that dominates the true posterior distribution \emph{for all mechanism outputs}. Because we're essentially paying for the heaviest elements in $\dd$, the bound in Theorem~\ref{thm:eps_bound_approx_comp} is tightest for data distributions that are close to uniform.

\begin{restatable}[Generalization of Theorem 5.2 \cite{steinke2024privacy}]{theorem}{thmfComp}\label{thm:eps_bound_f_comp}
Let \mech satisfy $f$-DP. Let dataset $X \sim \dd = \dd_1 \otimes \ldots \otimes \dd_n$ be drawn from a product distribution. Let $\adv$ be an attacker and let \metric be decomposable into $\ell_1 \cdots \ell_n$. Then, for all $\eps>0$, $v\in \R$,
\begin{align*}
    \Pr_{\substack{X \sim \dd \\ \mechout \sim \mech(X)}} [\metric(X, \adv(\mechout)) \geq v] \leq \Pr_{S^*_i \sim \bern(\beta_i(\eps))}\big[\sum_{i\in [n]} S^*_i \geq v\big] + \alpha \cdot n \cdot \delta_f(\eps), 
\end{align*}
where $\beta_i(\eps) = \frac{e^\eps}{e^\eps - 1 + \frac{1}{\Pr_{X\sim \dd_i}[\ell_i(X, \guesses^*) = 1]}}$ where $\guesses^*$ is the \textit{a priori} Bayes optimal attack attempt on \dd and  
\begin{align*}
    \alpha = \max\left\{\frac{1}{j} \left(\Pr_{S^*_i \sim \bern(\beta_i(\eps))}\left[\sum_{i\in [n]} S^*_i \geq v - j\right] -  \Pr_{S^*_i \sim \bern(\beta_i(\eps))}\left[\sum_{i\in [n]} S^*_i \geq v\right] \right) : j \in [n] \right\}.
\end{align*}
\end{restatable}
The approx-DP analog of this bound is \Cref{thm:eps_bound_approx_comp}.

To prove \Cref{thm:f_dp_non_comp}, \Cref{cor:f_basic}, and \Cref{thm:eps_bound_f_comp}, will use the fact that $f$-DP is known to be equivalent to an infinite collection of $(\eps, \delta_f(\eps))$-DP guarantees (Proposition 2.12~\cite{dong2019gaussian}) where $\delta_f(\eps) = 1 + f^*(-e^\eps)$ and $f^*$ is the convex conjugate of $f$. The upper bounds for the entire collection of $(\eps, \delta(\eps))$-DP guarantees will simply be satisfied simultaneously by an $f$-DP mechanism. Thus, the proofs for \Cref{thm:f_dp_non_comp}, \Cref{cor:f_basic}, and \Cref{thm:eps_bound_f_comp} follow immedately from their $(\eps, \delta)$-DP analogues, \Cref{thm:approx_non_comp}, \Cref{cor:approx_simple}, and \Cref{thm:eps_bound_approx_comp}, respectively.

\section{Idealized Gaussian Attack}
\label{app:idealized_gaussian}

The idealized setting for DP-SGD is a simple gaussian mechanism that calculates the noisy sum of orthogonal unit vectors, i.e., in Algorithm~\ref{alg:dist_recon}, each $\mathcal{D}_i$ is a distribution over $k$ $d$-dimensional orthogonal unit vectors $\{V^i_1, \dots, V^i_k\}$ and $a = \mathcal{M}(X_1, \dots, X_n) = \sum_i X_i + \mathcal{N}(0, \sigma^2I)$. In this case, the bayes optimal adversary $\mathcal{A}$ reconstructs each sample by choosing the vector with highest probability combined with the prior, i.e., $\mathbf{z}_i = \argmax_{V^i_j} \Pr[X_i = V^i_j | a] \propto \argmax_{V^i_j} \Pr[a | X_i = V^i_j] \cdot \Pr[V^i_j \sim \mathcal{D}]$. Lastly, the success of the adversary is exact match, i.e., $\mathcal{L}(X, \mathbf{z}) = \sum_i \I\{X_i = \mathbf{z}_i\}$.

In our experiments, since we are reconstructing a single digit, we set $k = 10$ and set $d = n * k$ where $n$ is the number of canaries we are reconstructing. Each $V^i_j$ is then a one-hot vector where the 1 is at index $i * k + j$. This ensures that the canaries both \emph{within} and \emph{between} each distribution $\mathcal{D}_i$ are orthogonal to each other. Note that this setting is functionally equivalent to an adversary with access to all intermediate models in DP-SGD in the case that the gradients of all possible canaries are orthogonal to each other (e.g., DP auditing~\cite{mahloujifar2024auditing,nasr2023tight}).

\section{Bounding Tabular Data Reconstruction from Noisy Marginals}
\label{sec:exps}
In this section, we investigate the tightness of our theoretical framework with respect to another commonly used DP mechanism.
Specifically, we compare our theoretical bounds against the success rate of a multi-attribute inference (M-AI) attack run against noisy $k$-way marginals, a commonly used DP mechanism in synthetic data.
In the process we also show how our framework can model novel privacy risks and attacks targeting multiple users.

\subsection{Experimental Setup}
\paragraph{$k$-way marginals.}
$k$-way marginals are useful summary statistics that effectively describe tabular datasets.
Specifically, a $k$-way marginal counts the number of records in the dataset whose values agree on $k$ attributes, e.g., ``How many people are \textbf{aged 40}, \textbf{unemployed} and \textbf{unmarried}'' is a 3-way marginal.

Although $k$-way marginals are useful summary statistics, they have been known to be vulnerable to powerful attacks~\cite{homer2008resolving,kasiviswanathan2010price,dick2023confidence}.
Therefore, calibrated Gaussian noise is typically added to these marginals so that they satisfy DP, which forms the basis for many popular DP synthetic data algorithms such as MST~\cite{mckenna2021winning}, AIM~\cite{mckenna2022aim}, and RAP~\cite{aydore2021differentially}.

\paragraph{Noisy $k$-way marginals.}
In our experiments, we assume the ``replace-one'' adjacency and additionally assume that for a given set of attributes, full marginal vectors across all possible values are released.
Therefore, we calibrate the level of Gaussian noise added to a $L_2$ sensitivity of $\sqrt{2}$ and $\delta = 10^{-5}$ for the entire marginal vector using Gaussian DP~\cite{dong2019gaussian} accounting.
Furthermore, as we are mainly focused on M-AI, we only calibrate the noise for the attributes that are \emph{unknown}, i.e., when $d'$ columns are unknown out of $d$ columns, the noise is calibrated to $\binom{d}{k} - \binom{d - d'}{k}$ compositions of the noisy $k$-way marginal mechanism.
Lastly, we fix $k = 3$ as these were the most powerful attacks in prior work~\cite{dick2023confidence}.
We provide a formal definition of the noisy $k$-way marginal mechanism below:

\begin{definition}[\textbf{$k$-way marginal query}]
Given a dataset $X$ over a data domain with $d$ attributes, $\mathcal{X} = \mathcal{X}_1 \times ... \mathcal{X}_d$, a $k$-way marginal query is defined by a subset of attributes $V \subseteq [d]$, $|V| = k$ and corresponding values for each of the attributes $v \in \prod_{i \in V} \mathcal{X}_i$. Given the pair $(V, v)$, define $\mathcal{X}(V, v) = \{x \in \mathcal{X} : x_i = v_i\; \forall i \in V \}$. The corresponding $k$-way marginal query is defined as follows where $\mathds{1}$ is the indicator function that maps elements of the set to 1 and the rest to 0:  
\begin{equation*}
    Q_{V, v}(X) = \frac{1}{|X|} \sum_{x \in X} \mathds{1} (x \in \mathcal{X}(V, v))
\end{equation*}
\end{definition}

\begin{definition}[\textbf{Noisy $k$-way marginal}]
    Given a dataset $X$ over a data domain with $d$ attributes, $\mathcal{X} = \mathcal{X}_1 \times \dots \times \mathcal{X}_d$. The noisy $k$-way marginal mechanism is defined by a subset of attributes $V \subseteq [d]$ and noise level $\sigma$ and outputs:
    \begin{equation*}
        \mathcal{M}_{V, \sigma}(X) = \left(Q_{V, v}(X) : \forall v \in \prod_{i \in V} \mathcal{X}_i\right) + \mathcal{N}(0, 2 \sigma^2I)
    \end{equation*}
\end{definition}

\paragraph{Adapting RAP attack.}
In prior work,~\citet{dick2023confidence} present a reconstruction attack against \emph{exact} $k$-way marginals by using the optimization objective introduced by~\citet{aydore2021differentially}, called the Relaxed Adaptive Projection (RAP) attack.
Specifically, they relax discrete datasets and the corresponding $k$-way marginals to the continuous domain and use standard optimization techniques such as Stochastic Gradient Descent to reconstruct the dataset using the following optimization objective: $\argmin_{\tilde{X'}} ||\tilde{Q}(\tilde{X}) - \tilde{Q}(\tilde{X'})||_2$ where $\tilde{Q}$, $\tilde{X}$, and $\tilde{X'}$ are the relaxed marginals, original dataset, and reconstructed dataset, respectively.
Finally, the relaxed reconstructed dataset is projected back to the discrete domain using rounding techniques.

In our experiments, we make three main modifications to what has been done in prior work and refer to this modified attack as our ``M-AI attack''.
First, we evaluate the effectiveness of the RAP attack against \emph{noisy} $k$-way marginals instead of exact.
By doing so, we measure the gap between our theoretical success bounds provided by DP and concrete attacks that can be run against DP mechanisms.
On the other hand,~\citet{dick2023confidence} focus on non-private mechanisms only.

Second, we adapt the RAP attack to perform \emph{multi-attribute inference} (M-AI) instead of reconstruction.
Unlike reconstruction, in M-AI the adversary is given access to a subset of attributes of the original dataset and tasked with reconstructing the remaining subset of attributes.
Naively, the reconstructed dataset $\tilde{X'}$ could be instantiated with these known attributes as the prior~\cite{dick2023confidence}.
However, we note that doing so does not guarantee that the final reconstructed dataset will be consistent with these known attributes, which might weaken the M-AI attack.
Therefore, along with instantiating $\tilde{X'}$ with the known attributes, during the optimization process, we additionally freeze these parameters so that the optimization process focuses purely on reconstructing the unknown attributes.

Lastly, instead of deriving the prior implicitly from other similar datasets, in our experiments we \emph{explicitly calculate the prior}.
Specifically, we fit a Bayesian Network (BayNet) model~\cite{koller2009probabilistic} to raw datasets and sample the ``original'' datasets from this model that the adversary then attempts to reconstruct.
We do so as this enables us to attack realistic looking datasets, while at the same time explicitly defining the prior distribution necessary for our bounds calculations.
Therefore, in our modified RAP attack, we initialize the unknown attributes of $\tilde{X'}$ to the explicit prior distribution conditioned on the known attributes calculated from the BayNet model, as opposed to initializing them to an adjacent dataset as done previously~\cite{dick2023confidence}.

\paragraph{Datasets.}
We experiment with two raw datasets in this work that have been commonly used for DP research~\cite{cai2021data,mckenna2021winning,mckenna2022aim}: American Community Survey (ACS)~\cite{acs2025} and San Francisco Fire Department Calls for Service (FIRE)~\cite{fire2022}.
Following the data preparation in prior work~\cite{annamalai2024linear}, we discretize and post-process the two datasets to have 16 and 10 attributes, respectively.
These datasets are then treated as the ``raw'' datasets used to fit a BayNet model, from which 100 fresh records are sampled as the ``original'' dataset.

\subsection{1 Column Reconstruction}
\begin{figure}[t]
    \centering
    \captionsetup[subfigure]{justification=centering}
    \subfloat[ACS]{
        \includegraphics[width=\linewidth]{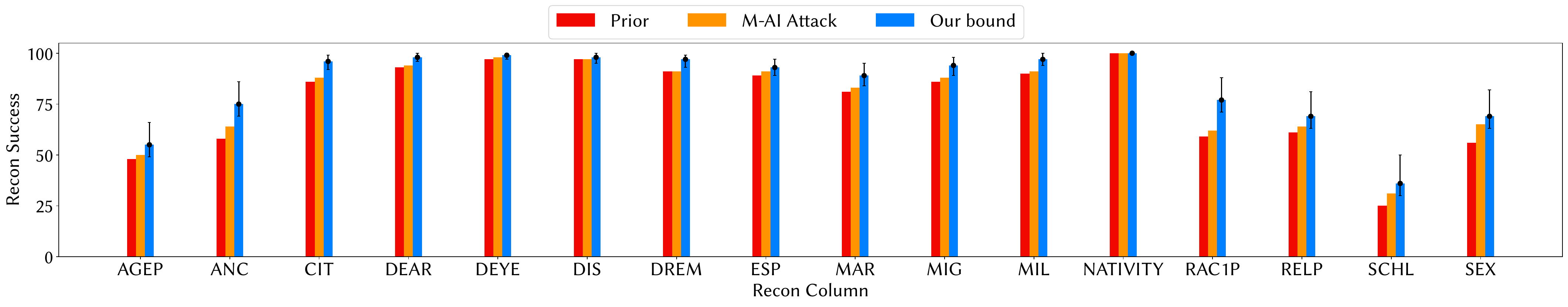}
        \label{fig:recon_1col_acs}
    }
    \\
    \subfloat[FIRE]{
        \includegraphics[width=\linewidth]{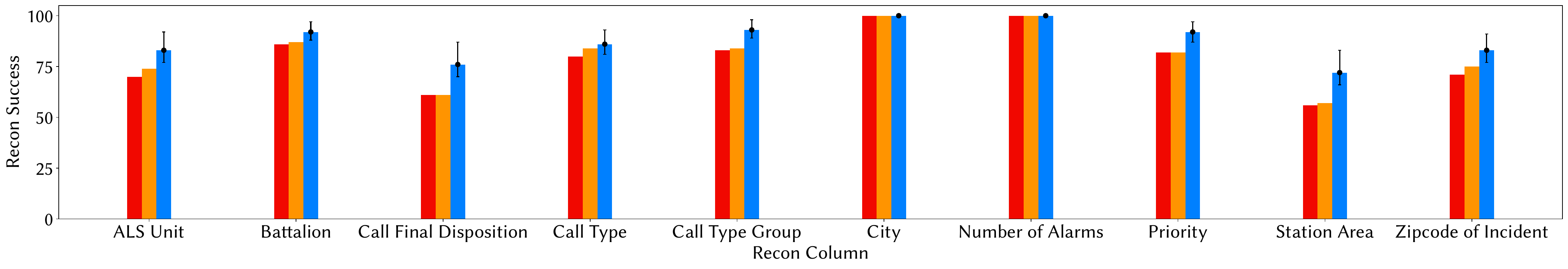}
        \label{fig:recon_1col_fire}
    }
    \caption{Success of reconstructing a single column at $\eps = 1$ and $\delta = 10^{-5}$ for different datasets and select columns.}
    \label{fig:recon_1col}
\end{figure}

We first begin with reconstructing a single column.
Note that this is very similar to traditional AI attacks, but we reconstruct the entire column for all records in the dataset, as opposed to AI attacks that typically only reconstruct a single attribute of a single record.
In Figure~\ref{fig:recon_1col} we report the total number of correctly reconstructed records when using our M-AI attack to reconstruct different columns from noisy $3$-way marginals of different datasets at a privacy level of $\eps = 1$.
Additionally, we compare the success of our attack with our theoretical upper bounds (derived from Theorem~\ref{thm:eps_bound_f_comp}) at varying levels of confidence (5\%, 50\%, 95\%) as error bars and the success of attempting to reconstruct the column from the prior distribution (derived from the BayNet model) alone.

First, we find that in such real-world settings the bounds can sometimes be vacuous on simple tasks such as reconstructing a single attribute.
Specifically, even for a high privacy level ($\eps = 1$), our theoretical bounds predict that most of the columns can be reconstructed almost fully ($\geq 80 / 100$) with a high level of confidence (95\%), with the exception of `AGEP' and `SCHL' columns of the ACS dataset.
This is mainly because real-world priors can be very large when reconstructing a single column.
Hence, we echo the concerns of~\citet{jayaraman2022attribute}, as care has to be taken when evaluating privacy attacks with respect to large priors in real-world datasets, which might seemingly lead to successful attacks even against DP mechanisms.
However, these attacks may not necessarily be \emph{significant} as they may not perform much better than the prior or anywhere close to the theoretical limit posed by DP.

Second, similar to our previous findings on data extraction, we note that real-world state-of-the-art attacks are still far from reaching our bound.
Although across both datasets and all columns, the success of our M-AI attack either matches or exceeds the success of using the prior alone, our M-AI attack is still far from realizing the limit of our theoretical bounds in many settings.
For instance, our M-AI attack achieves $65\%$ reconstruction success for the `SEX' attribute in the ACS dataset compared to the prior, which only achieves $56\%$.
However, our 95\% confidence level upper bound on the reconstruction success is $84\%$.

\begin{figure}[t]
    \centering
    \captionsetup[subfigure]{justification=centering}
    \subfloat[ACS, `SCHL']{
        \includegraphics[width=0.3\linewidth]{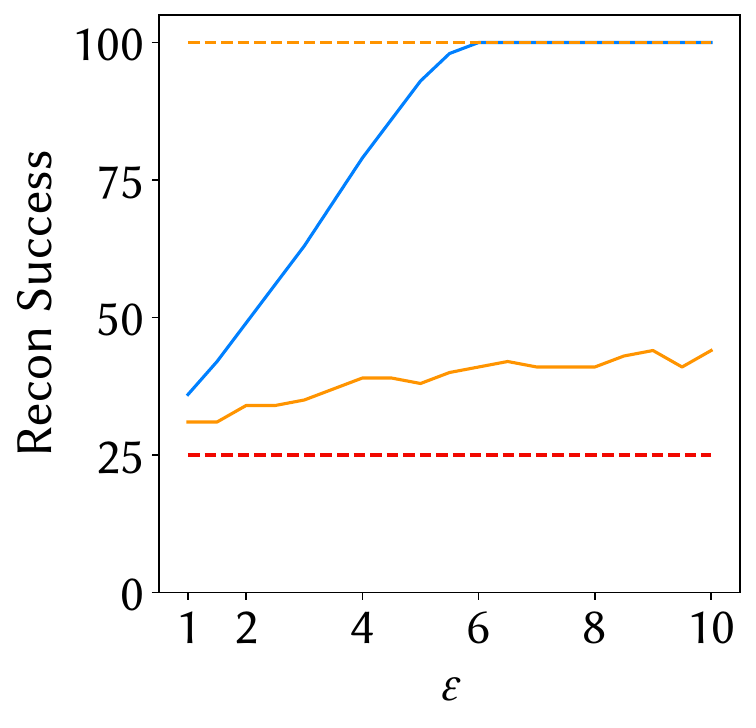}
        \label{fig:recon_1col_multeps_acs}
    }
    \subfloat[FIRE, `ALS Unit']{
        \includegraphics[width=0.3\linewidth]{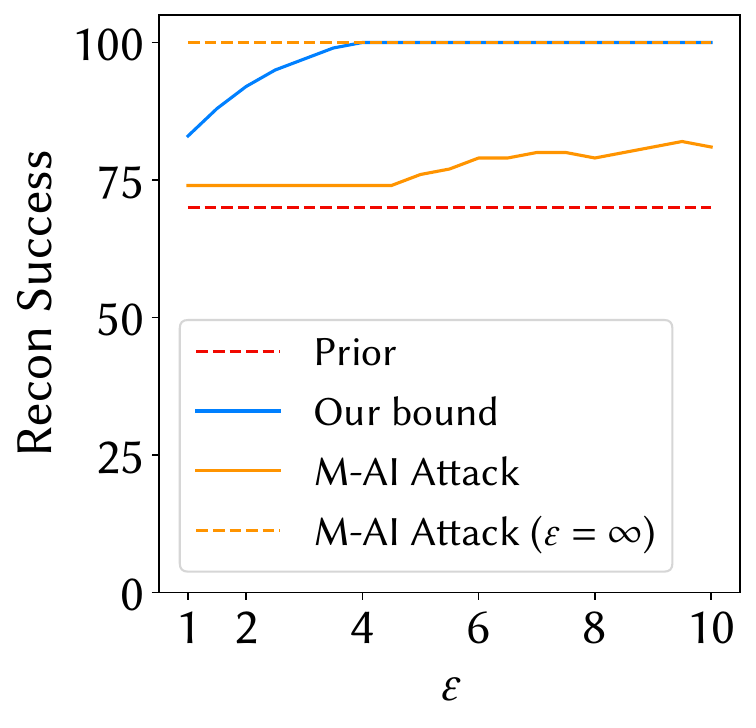}
        \label{fig:recon_1col_multeps_fire}
    }
    \caption{Success of reconstructing a single column at varying privacy levels $\eps$ and fixed $\delta = 10^{-5}$.}
    \label{fig:recon_1col_multeps}
\end{figure}

\paragraph{Impact of privacy level $\eps$.}
Next, in Figure~\ref{fig:recon_1col_multeps} we look at the impact of the privacy level $\eps$ on the reconstruction success of our bound and our M-AI attack.
To that end, we vary the $\eps$ and plot the corresponding success of reconstructing the `SCHL' and `ALS Unit' columns of the ACS and FIRE datasets, respectively.
Overall, as the $\eps$ increases, our attack becomes more significant, outperforming the prior more and becoming closer to our theoretical bounds.
In fact, for the vacuous privacy guarantee of $\eps = \infty$ our attack matches the theoretical bound and completely reconstructs the target columns of all records, whereas the prior success remains very low.
This indicates that, our M-AI attack is not fundamentally limited in its ability to leverage privacy leakage.
However, the success of state-of-the-art attacks that might work well against non-private mechanisms may not necessarily transfer to significant attacks for private mechanisms, especially at high levels of privacy.
Therefore, we believe that more research should focus on designing attacks specifically for DP mechanisms that can better leverage the privacy leakage and prior probabilities from these mechanisms specifically.

\subsection{Multi-Column Reconstruction}
\begin{figure}[t]
\centering
    \captionsetup[subfigure]{justification=centering}
    \subfloat[ACS]{
        \includegraphics[width=0.3\linewidth]{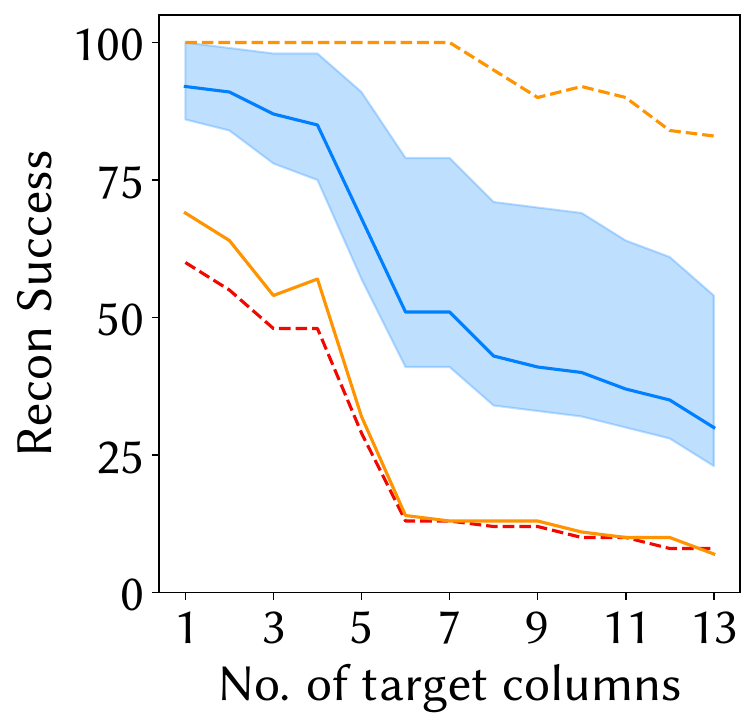}
        \label{fig:recon_multicol_acs}
    }
    \subfloat[FIRE]{
        \includegraphics[width=0.3\linewidth]{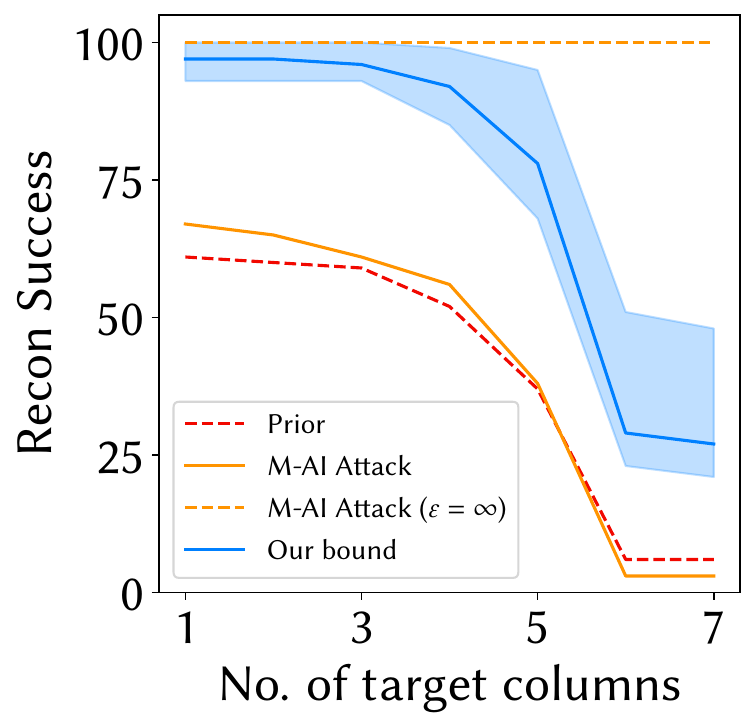}
        \label{fig:recon_multicol_fire}
    }
    \caption{Success of reconstructing multiple columns at $\eps = 3$ and $\delta = 10^{-5}$.}
    \label{fig:recon_multicol}
\end{figure}

We evaluate the success of our M-AI attack when reconstructing varying number of attributes and compare it with our theoretical bounds at fixed $\eps = 3$ in Figure~\ref{fig:recon_multicol}.
To that end, for each raw dataset, we fit a single BayNet model and reconstruct an increasing number of columns in the visit order of the network.
Note that in order to perform M-AI, a fixed number of columns must be known to match the reconstructed records with the original records, which we set to 3 in this case.
Therefore the maximum number of columns reconstructed for the ACS and FIRE datasets are 13 and 7, respectively.

Here, we observe that our M-AI attack remains most significant when the number of target columns reconstructed is low and quickly converges to the success of the prior.
For instance, for both datasets, we observe that when reconstructing more than 5 columns, our M-AI attack performs equivalently to just using the prior alone.
Furthermore, for the FIRE dataset, even though the theoretical bound stays roughly the same when reconstructing 3 columns as 1 column, the success of the M-AI attack reduces in tandem with the success of the prior.
This indicates that the M-AI attack is not optimal and largely depends on the prior as it struggles to reconstruct records as effectively under more difficult settings, even when the privacy leakage is the same.

\subsection{Approximate Reconstruction}
\begin{figure}[t]
\centering
    \captionsetup[subfigure]{justification=centering}
    \subfloat[ACS, `AGEP']{
        \includegraphics[width=0.3\linewidth]{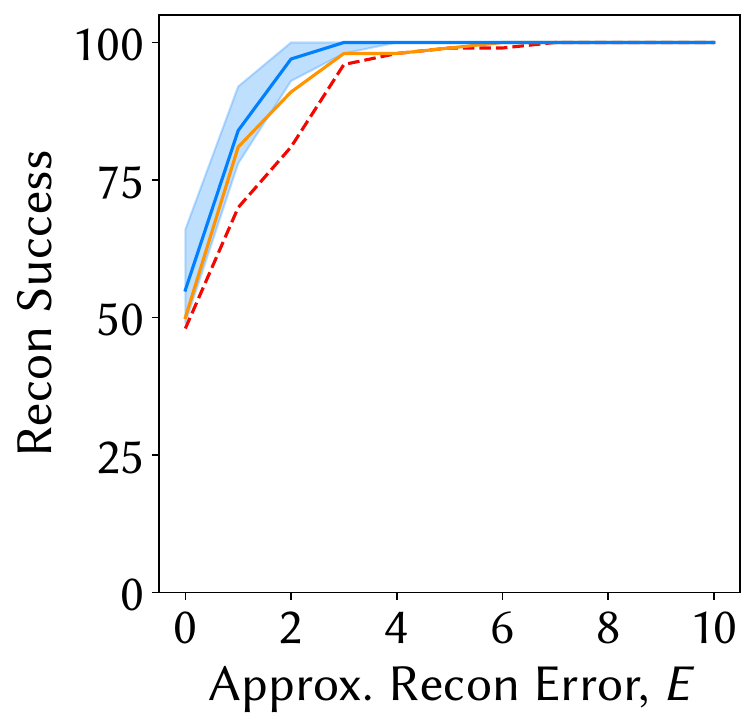}
        \label{fig:recon_multicol_approx_acs}
    }
    \subfloat[FIRE, `Zipcode of Incident']{
        \includegraphics[width=0.3\linewidth]{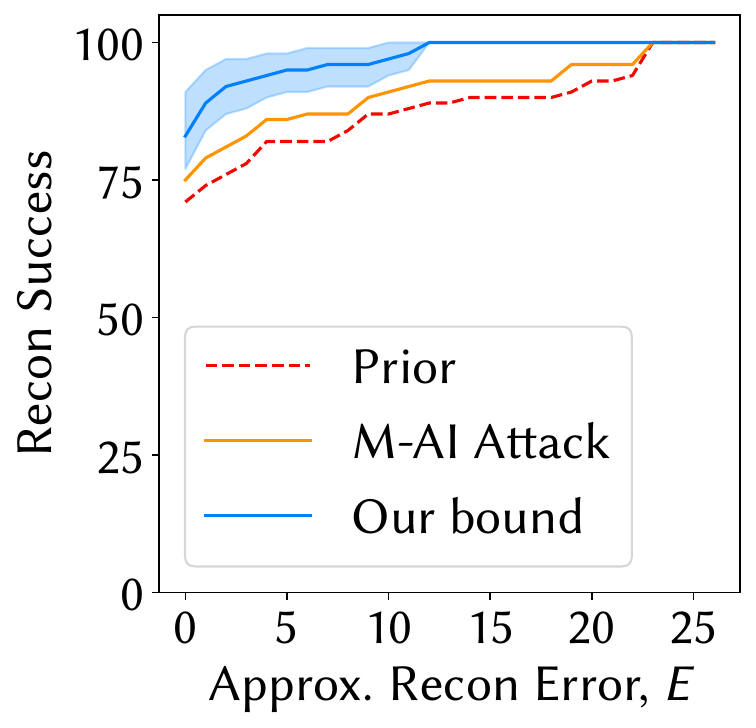}
        \label{fig:recon_multicol_approx_fire}
    }
    \caption{Success of approximately reconstructing a single column (up to $L_1$ error) at $\eps = 1$ and $\delta = 10^{-5}$.}
    \label{fig:recon_approx}
\end{figure}

Lastly, one of the key features of our framework is its ability to support both \emph{complex success metrics} and \emph{multiple targets} simultaneously, which none of the prior work has achieved before.
To that end, in this section, we focus on approximately reconstructing multiple records at the same time.

In Figure~\ref{fig:recon_approx} we report our reconstruction success when reconstructing the `AGEP' and `Zipcode of Incident' columns of the ACS and FIRE datasets, respectively, and compare it against our bounds when the success metric is defined as being within some $L_1$ error $E$, i.e., $\ell_i(X_i, \guesses) = \I\{|X_i - \guesses_i| \leq  E\}$.
Note that the `AGEP' and `Zipcode of Incident' columns have a domain of 10 and 26, respectively, which corresponds to the maximum reconstruction error possible in the Figure.
When $E = 0$, this corresponds to the ``exact match'' metric we have been using previously.

Here, we observe that even for approximate definitions of reconstruction, our attack consistently performs better than the prior, at times even matching our bound, albeit at very low levels of confidence.
This shows that not only can our framework successfully handle complex success metrics such as approximate reconstruction, but also that real world attacks might potentially perform better with respect to complex metrics as compared to simple metrics, which can be overly restrictive.

\subsection{Key Takeaways}
We show that while our bounds hold empirically, they remain loose when compared to the reconstruction success observed from actual attacks.
We believe that the main reason for this is that current state of the art reconstruction attacks may not be optimal for DP mechanisms, even though they work well under the non-private setting.
Nevertheless, in this section we show how our framework is capable of modelling \emph{novel privacy risks}, providing meaningful \emph{population-level} attack success bounds, and accommodating \emph{complex success metrics}, which have not been simultaneously captured by any of the prior work discussed. %
\section{Bridging Existing Attack Frameworks}
\label{sec:bridging_attack}
In this section, we show how our framework bridges and extends the current state of the art frameworks bounding the privacy leakage for a handful of influential works.

\subsection{Privacy Auditing in One Run} 

The influential work of Steinke, Nasr, and Jagielski \cite{steinke2024privacy} is perhaps the most closely related work to ours. At a high level, their work also gives a high probability bound on the success of an adversary with multiple targets, but only for MI. In Section~\ref{sec:encoding_attacks} we describe how to encode their specific MI setting in our bounds.
We generalize their bounds beyond MI attacks in two ways. First, we allow the data to be generated according to general product distributions rather than just sampling membership bits. Second, we generalize beyond measuring success via exact match by incorporating complex success metrics into the prior success probability.

\paragraph{Auditing $f$-Differential Privacy in One Run (MMC).}
Mahloujifar, Melis, and Chaudhuri~\cite{mahloujifar2024auditing} also derive MI bounds for $f$-DP mechanisms, which better captures common mechanisms like the Gaussian mechanism. They also generalize their membership inference bounds to what we call a ``multi-membership inference'' attack where each canary is drawn uniformly from a set of possible canaries, and the attacker is successful if they correctly guess the \emph{index} of the true canary out of the possible ones.

Multi-membership has a number of limitations in real-world applicability. In multi-membership inference, canaries have to have some ordering within their domain and attacks typically search over the potentially large domain of canaries to find the correct index. Furthermore, multi-membership inference cannot handle approximate reconstruction, since the attack success is measured by exactly matching the \emph{index} of the canary output by the adversary. On the other hand, in a full reconstruction attack, which is handled by our framework, the adversary outputs \emph{the canary itself} (or an approximation of it) that does not need any ordering and can be compared against original records using potentially approximate distance functions. Lastly, the MMC bound is not flexible enough to handle cases where there are more guesses than canaries, i.e., $k > n$.

 \subsection{Narcissus Resiliency} 

Perhaps the most interesting connection is that to the Narcissus Resiliency framework proposed by Cohen et al.~\cite{cohen2025data}. They define a general framework for evaluating a wide range of data reconstruction attacks and protections, including: Predicate singling out~\cite{cohen2020towards}, membership inference, differential privacy, one-way functions~\cite{diffie2022new, yao1982theory}, and various guarantees of encryption schemes. The core insight of their framework is to measure the attacker's success probability relative to its baseline probability of success on a fresh sample of data rather than using a fixed baseline for all attackers. A mechanism is considered ``narcissus resilient'' if the two success probabilities are close for all attackers. Specifically,

\begin{definition}[Narcissus resiliency, \cite{cohen2025data}]\label{def:narcissus_resil} Let \class be a class of data distributions and let \metric be a predicate over dataset, attack attempt pairs. Algorithm \mech is $(\eps, \delta, \class)$-\metric-Narcissus-resilient if for all $\dd \in \class$ and for all attackers $\adv$ it holds that
\begin{equation*}
    \Pr_{\substack{X \sim \dd \\ \mechout \gets \mech(X) \\ \guesses \gets \adv(\mechout)}} [\metric(X, \guesses) = 1] \leq e^\eps \cdot \Pr_{\substack{X \sim \dd\\ Y \sim \dd \\ \mechout \gets \mech(X) \\ \guesses \gets \adv(\mechout)}} [\metric(Y, \guesses) = 1] + \delta.
\end{equation*}
\end{definition}

We can directly compare our Corollary~\ref{cor:approx_simple} to the Narcissus Resiliency bounds for $(\eps, \delta)$-DP mechanisms. In Theorem~F.1, they bound the RHS by $e^\eps\cdot p + \delta$ where $p$ is the prior probability that predicate $\metric$ is true. This is always strictly larger than our bound of $\frac{e^\eps}{e^\eps + 1/p} + \delta$, especially when $p$ or $\eps$ are relatively large. In addition, we give high probability bounds that hold even when conditioned on the mechanism output, which describes the dependence on the number of successful attacks and the privacy parameters. Therefore, while our bounds enjoy the same philosophical benefits that inspired the narcissus resiliency framework, our bounds on DP mechanisms are substantially tighter.

\subsection{Reconstruction Robustness (ReRo)}

The reconstruction robustness (ReRo) game in \cite{balle2022reconstructing, hayes2024bounding} can be viewed as a specific case of our setting when there is only a single target, $n=1$, and a single reconstruction guess, $k=1$, and the probability is taken over the coins of the mechanism as well as the sampling of the target record. That is, translating the pure DP version of their bound (Corollary 3 in \cite{balle2022reconstructing}) into our notation, we get for $X'\sim \dd$
\begin{align}
    \E_{\mechout \sim \mech(X')}\left[\Pr_{\substack{X\sim \dd }}[\metric(X, R(\mechout))=1]\mid \mech(X) = \mechout \right]  \leq  e^\eps \cdot \sup_{z \in Supp(\dd)}\Pr_{X \sim \dd}[\metric(X, z)=1], \label{eq:rero}
\end{align}
where $\metric(X, z) = \I\{\ell(X, z) \leq \eta\}$ is some approximate reconstruction measure. In the ReRo game only one data record, $X$, is unknown while the rest are fixed and known to the adversary (one can imagine that \mech has the known records hardcoded).

Notice with these specific choices, the LHS is equal to 
$\Pr_{X\sim \dd, A \gets \mech(X)}(\ell(X, R(A)) < \eta)$ and the supremum on the RHS is exactly what their prior, denoted by $\kappa_{\ell, \dd}(\eta)$.

The ReRo bounds are a function of the prior success probability of the \emph{heaviest} item in the target's data distribution, which may be very far from the prior success probability of a real attack. %
As \citet{cohen2025data} point out, this can lead to cases where a sophisticated reconstruction of a relatively low-probability data record doesn't contradict the ReRo bound. Our bounds avoid this problem by computing the prior with respect to a specific attack attempt. In later work, Stock et al.~\cite{stock2022defending} derived a similar RDP reconstruction bound which has the same limitations as ReRo.

\paragraph{Distributional Reconstruction Robustness.} Cummings et al.~\cite{cummings2024attaxonomy} extend the pure-DP ReRo bounds to an attack setting with $n$ targets and one attack attempt ($k=1$). While they capture settings where the adversary has substantially less information about the dataset, their bounds suffer from the same problems as ReRo. Namely, the prior is still dependent on the heaviest item (the optimal a priori attack) and the bounds only give an expectation over mechanism outputs.

\subsection{Additional Related Work} Our work fits into a substantial body of previous work~\cite{balle2022reconstructing,stock2022defending,hayes2024bounding,steinke2024privacy,cherubin2024closed,cummings2024attaxonomy,mahloujifar2024auditing,kulynych2025unifying} that  bounds the success probability of particular attacks based on a DP guarantee.
These works consider several broad categories of attacks. \textit{Membership inference (MI)}~\cite{shokri2017membership}---in which an attacker given (partial) knowledge of a target individual and of the population from which the data are drawn aims to determine if the target individual's records was in the input data set---is the attack category that most closely tracks the negation of DP's guarantees.  \textit{Attribute inference (AI)}~\cite{ganju2018property} is similar but one aims to infer a specific sensitive attribute of the target, again given partial knowledge. \textit{Reconstruction}~\cite{dinur2003revealing} is the broadest class, encompassing multiple settings where the attacker infers one or more whole records.
In Table~\ref{tab:compare_prior_work} we compare our bounds to those of prior work.

While we focus here on works that aim to help interpret DP guarantees, many of the works we discuss also aim to check the quantitative tightness of the DP guarantee for a specific algorithm~\cite{steinke2024privacy,mahloujifar2024auditing}.%
\footnote{Such  assessment is often referred to as \textit{auditing}~\cite{ding2018detecting,jagielski2020auditing,nasr2023tight,steinke2024privacy,mahloujifar2024auditing,annamalai2025hitchhiker}. Mathematically, results on auditing overlap significantly with the type of bounds considered in this paper, but the goals and interpretation differ. Our comparison to previous work includes bounds motivated by both considerations.}

\end{document}